%% file: main.tex
\documentclass[letterpaper]{article} 
\usepackage[]{aaai2026}  

\usepackage{chrkroer}
\usepackage{defs}

\theoremstyle{plain}
\newtheorem{thm}{\protect\theoremname}
\theoremstyle{plain}

\theoremstyle{plain}
\newtheorem{cor}[thm]{\protect\corollaryname}

\makeatother

\providecommand{\lemmaname}{Lemma}
\providecommand{\theoremname}{Theorem}
\providecommand{\corollaryname}{Corollary}

\usepackage{times}  
\usepackage{helvet}  
\usepackage{courier}  
\usepackage[hyphens]{url}  
\usepackage{graphicx} 
\usepackage{natbib}  
\usepackage{caption} 
\usepackage{algorithm}
\usepackage{algorithmic}

\usepackage{comment}
\usepackage{booktabs}
\usepackage{longtable}

\usepackage{tikz}
\usepackage{pgfplots}
\usetikzlibrary{positioning,decorations.pathreplacing}

\usepackage{newfloat}
\usepackage{listings}
\DeclareCaptionStyle{ruled}{labelfont=normalfont,labelsep=colon,strut=off} 
\floatstyle{ruled}
\newfloat{listing}{tb}{lst}{}
\floatname{listing}{Listing}
\title{Spatial Branch-and-Bound for Computing Multiplayer Nash Equilibrium}
\author {
    Jakub \v{C}ern\'{y}\textsuperscript{\rm 1},
    Shuvomoy Das Gupta\textsuperscript{\rm 2},
    Christian Kroer\textsuperscript{\rm 1}
}
\affiliations {
    \textsuperscript{\rm 1}Columbia University, USA\\
    \textsuperscript{\rm 2}Rice University, USA\\
    jakub.cerny@columbia.edu, sd158@rice.edu, christian.kroer@columbia.edu
}

\usepackage{bibentry}

\begin{document}

\maketitle

\begin{abstract}
Equilibria of realistic multiplayer games constitute a key solution concept both in practical applications, such as online advertising auctions and electricity markets, and in analytical frameworks used to study strategic voting in elections or assess policy impacts in integrated assessment models. However, efficiently computing these equilibria requires games to have a carefully designed structure and satisfy numerous restrictions; otherwise, the computational complexity becomes prohibitive. In particular, finding even approximate Nash equilibria in general-sum normal-form games with two or more players is known to be PPAD-complete. Current state-of-the-art algorithms for computing Nash equilibria in multiplayer normal-form games either suffer from poor scalability due to their reliance on non-convex optimization solvers, or lack guarantees of convergence to a true equilibrium. In this paper, we propose a formulation of the Nash equilibrium computation problem as a polynomial complementarity problem and develop a complete and sound spatial branch-and-bound algorithm based on this formulation. We provide a qualitative analysis arguing why one should expect our approach to perform well, and show the relationship between approximate solutions to our formulation and that of computing an approximate Nash equilibrium. Empirical evaluations demonstrate that our algorithm substantially outperforms existing complete methods.
\end{abstract}

\begin{links}
    \link{Code}{github.com/CoffeeAndConvexity/SBnBForNE}
    \link{Extended version}{arxiv.org/abs/2508.10204}
\end{links}

\input{sections/intro}
\input{sections/prelims}
\input{sections/pipeline}
\input{sections/exps}
\input{sections/conclusion}

\section*{Acknowledgments}
This research was supported by the Office of Naval Research awards N00014-22-1-2530 and N00014-23-1-2374, and the National Science Foundation awards IIS-2147361 and IIS-2238960. We thank Mo\"{i}se Blanchard for his assistance with pure Nash equilibrium computation vectorization.

\bibliography{aaai2026}

\newpage
\onecolumn
\appendix
\input{sections/checklist}
\input{sections/appendix}

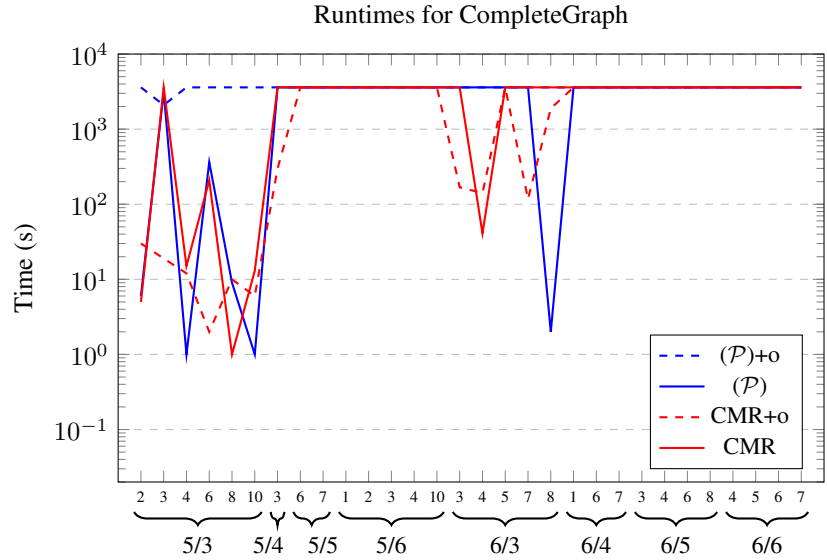
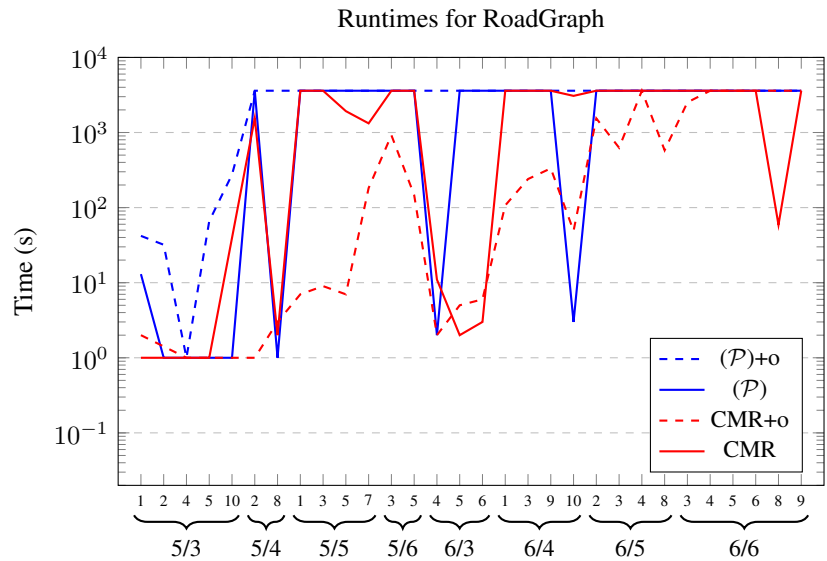
\end{document}

%% file: sections/intro.tex
\section{Introduction}

Computing equilibria in games is a long-standing research topic in artificial intelligence and has served as an important benchmark domain. Advances in equilibrium-finding algorithms enabled breakthroughs in AI for recreational settings such as poker, Stratego, and Diplomacy~\cite{bowling2015heads,brown2018superhuman,perolat2022mastering,meta2022human}, and have enabled decision-support systems for security and logistics applications~\citep{pita2008deployed,fang2015security,tambe2011security,jain2010software,cerny2024contested,cerny2025contested}. Deployed systems have helped allocate security resources across the eight terminals of Los Angeles International Airport~\citep{pita2008deployed}, support biodiversity protection over large conservation areas\citep{fang2015security,fang2016deploying,fang2016green}, and assist in screening millions of US air passengers annually\citep{brown2016one}. Related models have also been explored for traffic monitoring, drug interdiction, and cybersecurity~\citep{sinha2018stackelberg}. These applications depend on algorithms capable of computing game-theoretic equilibria at scale. Achieving such scale in practice typically requires modeling restrictions, most commonly reducing interactions to two players and, in many cases, assuming zero-sum utilities, which allows domain-agnostic methods to solve games efficiently.

While some interactions can be reasonably approximated as two-player, e.g., when dependencies between participants are limited or can be decomposed, many real-world settings inherently involve multiple interacting agents and must be modeled as multiplayer games. This is particularly common in computational economics, where various equilibrium concepts underpin deployed systems. Examples include Bayesian Nash equilibria in online advertising auctions~\cite{edelman2007internet}, market-clearing equilibria in electricity markets~\cite{kirschen2018fundamentals}, competitive equilibria in matching platforms such as ride-sharing~\cite{banerjee2015pricing,bimpikis2019spatial}, and Nash equilibria in spectrum auctions and telecommunications markets~\cite{cramton1997fcc,cramton2006combinatorial}. These scenarios are fundamentally multiplayer in nature. Solving them at scale typically requires even stronger modeling assumptions and significant exploitation of domain-specific structure to design tractable algorithms.

Relaxing these assumptions to develop more generally applicable and domain-agnostic algorithms for multiplayer settings remains a central challenge. This is especially true in contexts beyond traditional economic markets, such as multiplayer card games, network-based games like pursuit-evasion or interdiction, or security scenarios involving more than two self-interested entities, such as coalition/team-based games in organizations like the UN or NATO. Even in one of the most studied models -- finite normal-form games -- approximating a general-sum Nash equilibrium becomes PPAD-complete with two or more players, while two-player zero-sum games admit polynomial-time solutions.

Still, computing Nash equilibria in multiplayer normal-form games remains a traditional problem that has been studied for decades, and many classical algorithms exist. Support enumeration methods, such as full support enumeration and vertex enumeration, rely on exhaustively checking combinations of pure strategy supports and solving the resulting best-response constraints, typically expressed as systems of linear inequalities~\cite{porter2008simple,lipton2004nash}. Homotopy-based methods, including the Govindan-Wilson algorithm~\cite{govindan2003global}, simplicial subdivision~\cite{van1987simplicial}, and extended Lemke-type path-following algorithms~\cite{lemke1964equilibrium,mckelvey1996computation,turocy2005dynamic}, trace continuous solution paths from a perturbed game to the original game by following parameterized trajectories of fixed points. Polynomial system solvers approach the problem by encoding the equilibrium conditions, such as the Karush-Kuhn-Tucker conditions for each player’s optimization problem, as systems of polynomial equalities and inequalities, then solving them using Grobner bases, homotopy continuation, or algebraic elimination techniques~\cite{datta2010finding,mckelvey1996computation}. See~\citet{berg2017exclusion} for a more detailed overview.

Classical algorithms often suffer from numerical instability, particularly those based on solving polynomial systems, which can be sensitive to input perturbations and rounding errors. Their computational complexity also grows rapidly with the number of players and strategies, making them difficult to scale and limiting their applicability to relatively small games. Recent complete algorithms instead rely on formulations amenable to highly optimized industrial solvers such as CPLEX and Gurobi. These approaches use techniques such as decomposing multilinear terms into binary variables~\cite{ganzfried2024fast} or adding correlation plan constraints to restrict the feasible solution space~\cite{zhang2023computing}. An advantage of this approach is the ability to incorporate objectives into the formulation; for instance, to compute equilibria that maximize social welfare. Another modern approach is the exclusion method, a tree search-based algorithm that iteratively partitions the continuous solution space~\cite{berg2017exclusion}. Among incomplete methods, one of the most scalable is the gradient-based ADIDAS algorithm, which combines homotopy methods with iterative polymatrix solvers and can handle games with billions of outcomes~\cite{gemp2022adidas}. However, due to the use of biased gradient descent, ADIDAS lacks convergence guarantees to an exact equilibrium and instead offers local convergence only.

\paragraph{Our contributions.} We propose that the Nash equilibrium problem in multiplayer normal-form games should be solved as a polynomial complementarity problem. This formulation allows us to design a sound and complete algorithm based on spatial branch-and-bound. We provide qualitative analysis arguing why the reformulated problem is more tractable, and prove that a suboptimal solution to our reformulation yields an approximate Nash equilibrium, with an explicit bound on the approximation quality. Empirically, we demonstrate that our algorithm not only outperforms existing complete methods, but also remains competitive with the leading incomplete method on mid-sized games when run with early termination, while offering better-quality solutions. Large games still remain out of reach primarily due to the size of the utility matrix, which cannot be explicitly constructed. However, because our algorithm can serve as a subgame solver within incremental oracle-based methods, we see a path toward improving scalability in future work.

%% file: sections/prelims.tex
\section{Preliminaries}

We represent a multiplayer normal-form game as a tuple $G = (N, A, u)$, where $N$ denotes a finite set of players with $n = |N|$. For each player $i \in N$, let $A_i$ be the finite set of available actions, and define the joint action space as the Cartesian product $A = \prod_{i \in N} A_i$. A real-valued utility function $u_i : A \rightarrow \mathbb{R}$ specifies the payoff for player $i$ under each action profile $a \in A$, and we denote by $\minu_i$ and $\maxu_i$ the minimal and maximal values of $u_i$, respectively. Mixed strategies are denoted by $\Delta = \prod_{i \in N} \Delta_i$, where each $\Delta_i$ is the probability simplex over $A_i$. For any tuple indexed over players, we use the subscript $-i$ to indicate the components corresponding to all players except $i$; for example, for $a = (a_1, \dots, a_n) \in A$, we write $a_{-i} = (a_1, \dots, a_n) \setminus \{a_i\}$, and for a mixed strategy profile $\delta \in \Delta$, we write $\delta_{-i}(a_{-i}) = \prod_{j \in N, j \neq i} \delta_j(a_j)$, where $\delta_j \in \Delta_j$. The expected utility of player $i$ under a mixed strategy profile $\delta \in \Delta$ is given by
\begin{equation*}
u_i(\delta) = \sum_{a \in A} \left( \prod_{j \in N} \delta_j(a_j) \right) u_i(a),
\end{equation*}
and the expected utility of player $i$ when they play a specific action $a_i \in A_i$ against $\delta_{-i}$ is
\begin{equation*}
u_i(a_i, \delta_{-i}) = \sum_{a_{-i} \in A_{-i}} \delta_{-i}(a_{-i}) u_i(a_i, a_{-i}).
\end{equation*}
A \emph{Nash equilibrium} is a strategy profile $\delta = (\delta_i)_{i \in N} \in \Delta$ such that no player can improve their expected utility by unilaterally deviating from their strategy: formally,
\begin{equation}\tag{NE}
u_i(\delta_i, \delta_{-i}) \geq u_i(\delta_i', \delta_{-i}) \quad \forall\delta_i' \in \Delta_i.
\end{equation}
Finding a Nash equilibrium can be formulated as the following mixed-integer polynomial feasibility problem~\cite{shoham2008multiagent}:
\begin{equation}
\left(
\begin{array}{c}
\begin{aligned}
v_i &= s_i(a_i) + u_i(a_i, \delta_{-i})              &&\forall i, a_i \\
\delta_i(a_i) &\leq z_i(a_i)                         &&\forall i, a_i \\
s_i(a_i) &\leq (1 - z_i(a_i))(\maxu_i - \minu_i)     &&\forall i, a_i \\
\smashoperator{\sum_{a_i\in A_i}} \delta_i(a_i) &= 1                        &&\forall i \\
0 \leq \delta_i(a_i) &\leq 1                         &&\forall i, a_i \\
0 \leq s_i(a_i) &\leq \maxu_i - \minu_i              &&\forall i, a_i \\
\minu_i \leq v_i &\leq \maxu_i                       &&\forall i \\
z_i(a_i) &\in \{0,1\}                                &&\forall i, a_i
\end{aligned}
\end{array}
\right)
\tag{$\mathcal{P}$}\label{eq:miqcp-NE}
\end{equation}
where $\{\delta_i(a_i)\}_{i \in N, a_i \in A_i}$, $\{v_i\}_{i \in N}$, $\{z_i(a_i)\}_{i \in N, a_i \in A_i}$, $\{s_i(a_i)\}_{i \in N, a_i \in A_i}$ are the decision variables. Moreover, we call a strategy profile $\delta\in\Delta$ an $\epsilon$-equilibrium  if for all $i\in N$
\begin{equation}\tag{$\epsilon$-NE}\label{eps-nash}
u_i(\delta_i, \delta_{-i}) \geq u_i(\delta_i', \delta_{-i}) - \epsilon \quad \delta_i' \in \Delta_i.
\end{equation}

%% file: sections/pipeline.tex
\section{Our algorithm}

This section is organized as follows. First, we propose a
continuous formulation of \eqref{eq:miqcp-NE} as a polynomial
feasibility problem. Then, we discuss why a spatial branch-and-bound algorithm working with this continuous formulation can be significantly more tractable numerically than the original mixed-integer formulation. Finally,
we present a customized spatial branch-and-bound algorithm that exploits
the underlying structure of our continuous formulation to solve it
more efficiently.

\subsection{Equivalent continuous formulation of \eqref{eq:miqcp-NE}}

The equivalent continuous formulation of \eqref{eq:miqcp-NE} is as
follows: 
\begin{equation}
\left(\begin{array}{c}
\begin{aligned}v_{i}-u_{i}(a_{i},\delta_{-i}) & \geq0 &  & \forall i,a_{i}\\
\delta_{i}(a_{i})\left(v_{i}-u_{i}(a_{i},\delta_{-i})\right) & =0 &  & \forall i,a_{i}\\
\sum_{a_{i}\in A_{i}}\delta_{i}(a_{i}) & =1 &  & \forall i\\
0\leq\delta_{i}(a_{i}) & \leq1 &  & \forall i,a_{i}\\
\minu_{i}\leq v_{i} & \leq\maxu_{i} &  & \forall i
\end{aligned}
\end{array}\right)\tag{\ensuremath{\mathcal{Q}}}\label{eq:cont-NE}
\end{equation}
where the decision variables are now $\{\delta_{i}(a_{i})\}_{i\in N,a_{i}\in A_{i}}$,
$\{v_{i}\}_{i\in N}$, which is a strict subset of the decision variables
of \eqref{eq:miqcp-NE}. The second constraint, which states that
the product of two nonnegative terms must be zero, is known as a complementarity
constraint. We now show that the continuous formulation \eqref{eq:cont-NE}
and the mixed-integer formulation \eqref{eq:miqcp-NE} are equivalent
to each other.

\begin{proposition}\label{prop:P_eq_Q} The continuous formulation
\eqref{eq:cont-NE} and the mixed-integer polynomial formulation \eqref{eq:miqcp-NE}
are equivalent. \end{proposition}

\begin{proof} First, assume we have a feasible solution $(\delta,v,z,s)$
for \eqref{eq:miqcp-NE}, then we can show that $(\delta,v)$ is a
feasible solution to \eqref{eq:cont-NE}. The third and fourth constraints
$\sum_{a_{i}\in A_{i}}\delta_{i}(a_{i})=1$, $0\le\delta_{i}(a_{i})\le1$,
and $\minu_{i}\le v_{i}\le\maxu_{i}$ are identical in both formulations,
so $(\delta,v)$ will satisfy those trivially. For the first constraint
of \eqref{eq:cont-NE}, note that from \eqref{eq:miqcp-NE}, we have
$s_{i}(a_{i})=v_{i}-u_{i}(a_{i},\delta_{-i})$, and $s_{i}(a_{i})\ge0$,
which implies $v_{i}-u_{i}(a_{i},\delta_{-i})\geq0$. To show that
the second constraint of \eqref{eq:cont-NE} will be satisfied by
$(\delta,v)$, consider two cases: (i) $\delta_{i}(a_{i})=0$, where
the second constraint is trivially satisfied (ii) $\delta_{i}(a_{i})>0$,
where $\delta_{i}(a_{i})\le z_{i}(a_{i})$ from \eqref{eq:miqcp-NE}
forces $z_{i}(a_{i})=1$ implying $s_{i}(a_{i})\le(1-z_{i}(a_{i}))(\maxu_{i}-\minu_{i})=0$,
and because $s_{i}(a_{i})\ge0$ in \eqref{eq:miqcp-NE}, we must have
$s_{i}(a_{i})=v_{i}-u_{i}(a_{i},\delta_{-i})=0$, thus satisfying
the second constraint of \eqref{eq:cont-NE} as well. Hence, $(\delta,v)$
is a feasible solution to \eqref{eq:cont-NE}.

Next, assume we have a feasible solution $(\delta,v)$ for the continuous
problem \eqref{eq:cont-NE}. Then we can construct some $(z,s)$ such
that $(\delta,v,z,s)$ is a feasible solution to \eqref{eq:miqcp-NE},
which we can show as follows. Define $s_{i}(a_{i}):=v_{i}-u_{i}(a_{i},\delta_{-i})$
(thus automatically satisfying the first constraint of \eqref{eq:miqcp-NE})
and $z_{i}(a_{i})\coloneqq1$ if $\delta_{i}(a_{i})>0$ and $z_{i}(a_{i})\coloneqq0$
if $\delta_{i}(a_{i})=0$ (which satisfies the last constraint of
\eqref{eq:miqcp-NE} via construction). For the second constraint
of \eqref{eq:miqcp-NE}, we consider two cases: (i) if $\delta_{i}(a_{i})>0$,
then $z_{i}(a_{i})=1$ by construction and because $\delta_{i}(a_{i})\le1$,
it implies $\delta_{i}(a_{i})\leq z_{i}(a_{i})$ holds, (ii) If $\delta_{i}(a_{i})=0$,
then $z_{i}(a_{i})=0$ by construction satisfying $\delta_{i}(a_{i})\leq z_{i}(a_{i})$.
The constraint $s_{i}(a_{i})\ge0$ is guaranteed by our construction
and $s_{i}(a_{i})\le\maxu_{i}-\minu_{i}$ is always true as $v_{i}\le\maxu_{i}$
and $u_{i}(a_{i},\delta_{-i})\ge\minu_{i}$. For $s_{i}(a_{i})\leq(1-z_{i}(a_{i}))(\maxu_{i}-\minu_{i})$,
we consider two cases again: (i) if $\delta_{i}(a_{i})>0$, then $z_{i}(a_{i})=1$
by construction and the second constraint of \eqref{eq:cont-NE} implies
$v_{i}-u_{i}(a_{i},\delta_{-i})=0$ in this case so $s_{i}(a_{i})=0$
thus $0=s_{i}(a_{i})\leq(1-z_{i}(a_{i}))(\maxu_{i}-\minu_{i})=0$,
(ii) If $\delta_{i}(a_{i})=0$, then $z_{i}(a_{i})=0$ by construction,
leading to $s_{i}(a_{i})\le(1-z_{i}(a_{i}))(\maxu_{i}-\minu_{i})=\maxu_{i}-\minu_{i}$,
and we showed earlier $s_{i}(a_{i})\le\maxu_{i}-\minu_{i}$ is true.
Hence, any solution to \eqref{eq:cont-NE} implies a solution to \eqref{eq:miqcp-NE}.
\end{proof}

To solve \eqref{eq:cont-NE} using the spatial branch-and-bound algorithm
effectively, we work with the following penalized version of \eqref{eq:cont-NE},
where the complementarity constraint of \eqref{eq:cont-NE} is penalized
through the $\infty$-norm via the epigraph approach \cite[page 134]{boyd2004convex}
as:

\begin{equation}
\left(\begin{aligned}\text{min}\quad & \varpi\\
\text{s.t.}\quad & \varpi\geq\delta_{i}(a_{i})\left(v_{i}-u_{i}(a_{i},\delta_{-i})\right) &  & \forall i,a_{i}\\
 & \varpi\geq-\delta_{i}(a_{i})\left(v_{i}-u_{i}(a_{i},\delta_{-i})\right) &  & \forall i,a_{i}\\
 & v_{i}-u_{i}(a_{i},\delta_{-i})\geq0 &  & \forall i,a_{i}\\
 & \sum_{a_{i}\in A_{i}}\delta_{i}(a_{i})=1 &  & \forall i\\
 & 0\leq\delta_{i}(a_{i})\leq1 &  & \forall i,a_{i}\\
 & \minu_{i}\leq v_{i}\leq\maxu_{i} &  & \forall i
\end{aligned}
\right)\tag{\ensuremath{\mathcal{R}}}\label{eq:cont-NE-penal}
\end{equation}
where the decision variables are $\{\delta_{i}(a_{i})\}_{i\in N,a_{i}\in A_{i}}$,
$\{v_{i}\}_{i\in N}$ and the scalar epigraph variable $\varpi$.
We now show that any solution to \eqref{eq:cont-NE} is also a solution
to \eqref{eq:cont-NE-penal} and vice versa.

\begin{proposition}\label{prop:Q_eq_R} Problems \eqref{eq:cont-NE}
and \eqref{eq:cont-NE-penal} are equivalent. \end{proposition}

\begin{proof}

Note that the constraints of \eqref{eq:cont-NE-penal} are the same
as \eqref{eq:cont-NE}, except the complementarity constraint, which
is penalized via the epigraph approach \cite[page 134]{boyd2004convex}
with objective effectively minimizing the nonnegative objective $\max_{i\in N,a_{i}\in A_{i}}|\delta_{i}(a_{i})\left(v_{i}-u_{i}(a_{i},\delta_{-i})\right)|$.
Because every finite game admits at least one Nash equilibrium \cite{nash1951non},
there exists an optimal solution to \eqref{eq:cont-NE}, and by construction,
such a Nash equilibrium will be a feasible solution to \eqref{eq:cont-NE-penal}
with objective value $0$, hence it will also be the globally optimal
solution to \eqref{eq:cont-NE-penal} because $0$ is the lowest possible
value the objective of \eqref{eq:cont-NE-penal}.

Similarly, due to \cite{nash1951non}, an optimal solution to \eqref{eq:cont-NE-penal}
will have objective value $0$ and such a solution will satisfy all
the constraints shared by \eqref{eq:cont-NE-penal} and \eqref{eq:cont-NE}.
Because the objective $\max_{i\in N,a_{i}\in A_{i}}|\delta_{i}(a_{i})\left(v_{i}-u_{i}(a_{i},\delta_{-i})\right)|$
being zero is equivalent to each complementarity constraint being
satisfied, a solution to \eqref{eq:cont-NE-penal} thus satisfies
every constraints of \eqref{eq:cont-NE}. \end{proof} Proposition
\ref{prop:Q_eq_R} leads to the following corollary: \begin{cor}
\label{cor:opt_val_R}The optimal objective value of \eqref{eq:cont-NE-penal}
is zero and any optimal solution to \eqref{eq:cont-NE-penal} is an
exact Nash equilibrium. \end{cor} %
Moreover, if any algorithm solving \eqref{eq:cont-NE-penal} terminates
with a feasible solution having a non-zero objective value, we have
the following guarantee regarding the solution being an $\epsilon$-Nash
equilibrium. 

\begin{proposition}\label{prop:approx} A candidate solution $\delta$
that satisfies all the constraints in \eqref{eq:cont-NE-penal} with
objective value $\varpi$ constitutes an $\epsilon$-Nash equilibrium
with $\epsilon=\left(\varpi\times\max_{i\in N}|A_{i}|\right)$. \end{proposition}
\begin{proof} For each player $i$, let $v_{i}^{*}$ be their best-response
value against the profile $\delta_{-i}$, i.e., $v_{i}^{*}=\max_{a_{i}\in A_{i}}u_{i}(a_{i},\delta_{-i})$.
The utility deviation improvement $\epsilon_{i}$ is then equal to
\[
\begin{aligned}\epsilon_{i} & =v_{i}^{*}-u_{i}(\delta)=v_{i}^{*}-\sum_{a_{i}\in A_{i}}\delta_{i}(a_{i})u_{i}(a_{i},\delta_{-i})\\
 & \leq v_{i}-\sum_{a_{i}\in A_{i}}\delta_{i}(a_{i})u_{i}(a_{i},\delta_{-i})\\
 & =\sum_{a_{i}\in A_{i}}\delta_{i}(a_{i})(v_{i}-u_{i}(a_{i},\delta_{-i}))\leq\varpi|A_{i}|,
\end{aligned}
\]
where $v_{i}^{*}\leq v_{i}$ follows from the third constraint in
\eqref{eq:cont-NE-penal}, and the final inequality uses the first
two constraints in \eqref{eq:cont-NE-penal}. Moreover, we use the
fact that $\sum_{a_{i}\in A_{i}}\delta_{i}(a_{i})=1$. The approximation
parameter of the equilibrium $\epsilon$ is then equal to $\epsilon=\max_{i\in N}\epsilon_{i}=\varpi\max_{i\in N}|A_{i}|$.
\end{proof} 

In the next section, we
discuss why for an effective spatial branch-and-bound implementation,
it is more convenient to work with \eqref{eq:cont-NE-penal} over
\eqref{eq:miqcp-NE}.

\subsection{Tractability of \eqref{eq:cont-NE-penal} over \eqref{eq:miqcp-NE}}

For a spatial branch-and-bound algorithm, working with the penalized continuous formulation \eqref{eq:cont-NE-penal} can be significantly more tractable than with the mixed-integer formulation \eqref{eq:miqcp-NE}.
This advantage stems from three primary factors: the dimensionality
of the search space, the complexity of the branching scheme, and the
ease of generating high-quality feasible solutions. To explain these
advantages clearly, we first provide a brief overview of how a standard
spatial branch-and-bound algorithm works, for more details we refer
the interested readers to \cite{liberti2008introduction,horst2013global,locatelli2013global,Gurobi91}.

\paragraph{Brief review of the standard spatial branch-and-bound algorithm.}

Consider a nonconvex mixed-integer polynomial minimization problem
where the objective and constraint functions are polynomials, and
the decision variables are a mix of continuous and binary types. We
assume that the problem's feasible region is nonempty and bounded
and that it admits a finite optimal value $p^{\star}$, as is the
case for our problem. The spatial branch-and-bound algorithm employs
a divide-and-conquer strategy to solve the problem by recursively
partitioning the domain to create a search tree. The process begins
with a \emph{presolve} phase, where a linear relaxation of the original
problem is solved to obtain valid bounds on the continuous variables.
Following this, the algorithm dynamically partitions the feasible
region into a finite collection of subregions, each corresponding
to a subproblem. By definition, the smallest among the optimal values
of these subproblems is the optimal value of the original problem.
The speed of the algorithm is essentially controlled by the efficiency
of solving these subproblems and, if necessary, branching them further
into smaller subproblems, which we discuss next.

While solving a particular subproblem---which is also a nonconvex
polynomial minimization problem, albeit over a smaller region---can
be as hard as solving the original problem, computing lower and upper
bounds for that subproblem via linear or convex relaxations is much
easier. At the root node of the spatial branch-and-bound tree, a linear
relaxation of the original problem is constructed and solved to yield
a lower bound on $p^{\star}$, denoted by $\underline{p}^{\star}$.
The tighter this relaxation, the closer $\underline{p}^{\star}$ is
to $p^{\star}$. In the absence of any special structure, the initial
upper bound on $p^{\star}$ is set to $\overline{p}^{\star}=\infty$.
However, a user can \emph{warm-start} the algorithm by providing a
known feasible solution, which can establish a nontrivial initial
$\overline{p}^{\star}$. Efficient warm-starting procedures that exploit
the problem structure can massively accelerate branch-and-bound solvers,
a technique we will leverage for problem \eqref{eq:cont-NE-penal}
(described in the next section). Note that a polynomial feasibility
problem, such as finding a Nash Equilibrium in \eqref{eq:miqcp-NE},
cannot be warm-started, as providing a feasible solution is equivalent
to having already solved it. Throughout its execution, the algorithm
updates $\underline{p}^{\star}$ and $\overline{p}^{\star}$ to progressively
reduce the optimality gap, $\overline{p}^{\star}-\underline{p}^{\star}$
and the algorithm terminates when $\underline{p}^{\star}$ matches
$\overline{p}^{\star}$, having found a globally optimal solution.

The algorithm reduces this gap by iteratively processing subproblems.
It selects an active subregion from the list of unprocessed nodes
in the search tree (typically guided by heuristics) and constructs
two convex (usually linear) minimization subproblems. These convex
formulations are created using convex lower and upper envelopes \cite{mccormick1976computability}
for the nonlinear terms, whereas the linear constraints and terms
remain unchanged. Additionally, valid linear inequalities (cuts) can
be added to strengthen these convex relaxations \cite{sherali1990hierarchy,padberg1989boolean,sherali2002enhancing}.
Solving these convex subproblems provides valid lower and upper bounds
on the optimal value for the active subregion, leading to one of the
following three possibilities: 
\begin{itemize}
\item \textbf{Possibility 1: pruning.} If the convex optimization problem
associated with the lower bound is infeasible or its objective value
is greater than the current global upper bound $\overline{p}^{\star}$,
then the subregion cannot contain the global optimum. Consequently,
this subregion is discarded, or \emph{pruned}, from the search tree. 
\item \textbf{Possibility 2: fathoming by optimality.} For the subproblem,
if its lower bound matches its upper bound%
, we have found this subproblem's optimal solution without solving
the original nonconvex problem on that subregion. This local solution
is also a feasible solution to the main problem. If its objective
value is better than the current $\overline{p}^{\star}$, it becomes
the new incumbent solution, and $\overline{p}^{\star}$ is updated
accordingly. The subregion is now fully explored, or \emph{fathomed},
and is not branched further. 
\item \textbf{Possibility 3: branching.} If neither of the first two possibilities
occurs, the current subregion may still contain the global optimum
and must be explored further. It is partitioned into smaller subregions
by performing \emph{branching}. The algorithm typically prioritizes
\emph{combinatorial branching} on a binary variable that has a fractional
value in the convex relaxation's solution. If all binary variables
are integer-valued, then it performs \emph{spatial branching} by selecting
a continuous variable involved in a violated nonlinear constraint.
These new, smaller subregions are then added to the list of active
nodes to be processed. 
\end{itemize}
This process is repeated, continually updating the global bounds.
The global upper bound, $\overline{p}^{\star}$, is always the objective
value of the best-known feasible solution i.e., the incumbent. The
global lower bound, $\underline{p}^{\star}$ is updated by taking
the minimum of the best objective values of all the leaf nodes associated
with the subproblems that are pruned or fathomed. As the algorithm
explores the search space, the gap $\overline{p}^{\star}-\underline{p}^{\star}$
progressively narrows, and when the lower bound $\underline{p}^{\star}$
matches the upper bound $\overline{p}^{\star}$, the algorithm terminates,
returning the incumbent solution as a global optimum for the original
problem.

\paragraph{Why \eqref{eq:cont-NE-penal} is more tractable than \eqref{eq:miqcp-NE}.}

We now discuss three key reasons why the penalty formulation \eqref{eq:cont-NE-penal}
is more tractable than the mixed-integer formulation \eqref{eq:miqcp-NE}
within a spatial branch-and-bound framework.

First, the set of decision variables in \eqref{eq:cont-NE-penal},
$\{\delta,v\}$, is a strict subset of the variables in \eqref{eq:miqcp-NE},
which additionally includes the auxiliary binary variables $z$ and
slack variables $s$. A spatial branch-and-bound algorithm's performance
is highly sensitive to the number of variables on which it must branch.
By operating in a lower-dimensional space, the formulation \eqref{eq:cont-NE-penal}
inherently defines a smaller and less complex search tree, mitigating
the curse of dimensionality and reducing the overall computational
effort.

Second, in \eqref{eq:cont-NE-penal}, the algorithm only needs to
perform spatial branching, which occurs exclusively on the continuous
variables $\delta$ and $v$ to resolve violations in the polynomial
complementarity constraints. In contrast, solving \eqref{eq:miqcp-NE}
requires a hybrid branching scheme involving both combinatorial branching
on the binary variables and spatial branching on the continuous variables.
This dual-purpose branching complicates the node selection logic and
increases the number of potential branches at each node of the tree.

Third, the formulation \eqref{eq:cont-NE-penal} is exceptionally
well-suited for computing high-quality feasible solutions using nonlinear
interior-point methods \cite{byrd1999interior,waltz2006interior,wachter2002interior}.
In theory, when warm-started with a good initial point, such algorithms
converge quadratically to a locally optimal solution for nonlinear
programs, where all functions are twice continuously differentiable
\cite{WachterBiegler2005local,WachterBiegler2005global,byrd2000trust};
empirically, they often perform well even when the warm-start point
is not high-quality \cite{dasgupta2022BnBPEP}. Every function in
\eqref{eq:cont-NE-penal} is twice continuously differentiable, and
a basic feasible solution can be computed trivially. Furthermore,
because minimizing the objective of \eqref{eq:cont-NE-penal} to zero
corresponds to satisfying the underlying complementarity conditions,
modern interior-point solvers that include special interfaces for
complementarity constraints can handle the problem structure with
particular efficiency \cite{KnitroComplementarity,byrd2006k}. The
ability to quickly find a locally optimal solution is a significant
advantage, as this solution provides a strong upper bound that allows
the branch-and-bound algorithm to prune large portions of the search
tree, thereby drastically improving its overall performance. Furthermore,
such a locally optimal solution is likely to have a small objective
value, thus corresponding to an $\epsilon$-Nash equilibrium by Proposition \ref{prop:approx}.

\subsection{Customized spatial branch-and-bound}

The computational tractability of \eqref{eq:cont-NE-penal} allows
us to apply the following customized two-stage spatial branch-and-bound
algorithm that exploits the problem structure: 
\begin{itemize}
\item \textbf{Stage 1.} We first compute a locally optimal solution to \eqref{eq:cont-NE-penal}
using a nonlinear interior point method, as discussed in the previous section.
From Proposition \ref{prop:approx}, we know that the objective value of \eqref{eq:cont-NE-penal} for a feasible solution tells us how close the found solution is to a Nash equilibrium.
Because of this, a locally optimal solution with a small objective value acts as an
excellent warm-starting point for the spatial branch and bound algorithm. 
\item \textbf{Stage 2.} We then feed this locally optimal solution as a
warm-starting point to the spatial branch-and-bound algorithm. Since
many sub-regions of the branch-and-bound tree will likely not contain
any Nash equilibrium, their convex relaxations will often yield objective
values strictly greater than zero. Hence a small upper bound provided
by a locally optimal solution ensures these sub-regions are quickly
identified and eliminated, thus aggressively pruning the vast majority
of the search space, drastically reducing the number of nodes to explore
in turn and leading to a significant reduction in computation time. 
\end{itemize}
Our discussion above results in the following proposition: 

\begin{proposition} The two-stage spatial branch-and-bound algorithm
is a sound and complete algorithm for finding a~Nash equilibrium
in multiplayer games. If we choose to terminate the algorithm early,
it returns an approximate Nash equilibrium along with its corresponding
approximation parameter.\end{proposition} 

\begin{proof} The algorithm\textquoteright s soundness follows from
Corollary~\ref{cor:opt_val_R} in the case of exact solutions, and
from Proposition~\ref{prop:approx} when terminated early. 
To see the completeness of our spatial branch-and-bound algorithm (i.e.
that it computes a solution if one exists), consider
the worst-case scenario, where the stage 1 is unable to provide a
point to warm-start the spatial branch-and-bound algorithm. In that
case, our customized spatial branch-and-bound algorithm reverts to
the standard spatial branch-and-bound algorithm, which is complete~\cite{horst2013global}.
\end{proof}

%% file: sections/exps.tex
\section{Empirical Evaluation}

\input{figs/runtimes_complete}

We empirically analyze the performance and scalability of our proposed branch-and-bound method, focusing on (i) wall-clock computational time, (ii) equilibrium approximation quality. While prior work has evaluated algorithms on fully random games~\cite{ganzfried2024fast,zhang2023computing,berg2017exclusion}, we do so only in the appendix, as we believe the resulting performance reports are not sufficiently informative. Random games frequently admit sparse equilibria, which may render them more tractable than instances with more structured or adversarial properties. For example, although computing a Nash equilibrium in two-player general-sum games is PPAD-hard~\citep{chen2006settling}, games with appropriately chosen random payoff distributions can often be solved in expected polynomial time~\citep{barany2007nash}. More broadly, in multi-player games with i.i.d. payoffs, the existence of pure Nash equilibria is almost certain, with the probability approaching $1 - 1/e$ as action set sizes increase~\citep{goldberg1968probability,dresher1970probability}. Similarly, two-player Gaussian or uniform games typically admit equilibria supported on just two actions with high probability, specifically, $1 - O(1/\log n)$~\citep{barany2007nash}.

Interestingly, we observed that similar characteristics appear even in games generated by the widely used GAMUT library~\cite{nudelman2004run}, despite these games incorporating various structural constraints rather than being entirely random. To investigate this, we implemented a simple vectorized method in NumPy to identify pure strategy profiles with minimal $\epsilon$ and examined the following GAMUT game classes: BidirectionalLEG, CollaborationGame, CovariantGame, RandomGraphicalGame, and RandomLEG. For games with up to 6 players and 6 actions per player (where applicable), we generated 30 random instances for each setting. Our method found pure profiles with least exploitability for all games within 2 seconds. Among the tested classes, only RandomGraphicalGame consistently generated instances with non-trivial equilibria. For this reason, in this paper we chose to focus our experimental evaluation on this class, using three types of underlying graphs: CompleteGraph, RoadGraph, and SmallWorldGraph, with the graph parameter setting \texttt{-reflex ok 0}. All utilities were normalized to the $[0, 1]$ range.

We conducted all the experiments on a computer running Ubuntu 20.04.6 LTS on an Intel(R) Core(TM) i9-12900 CPU with 16 cores and 24 threads with 64 GB of RAM. We used \texttt{JuMP} \cite{Lubin2023}, a domain specific modeling language embedded in the open-source programming language \texttt{Julia} \cite{Julia-2017}, to model our optimization problems, and  our customized two-stage spatial branch-and-bound algorithm uses the solvers \texttt{KNITRO 14.2} \cite{byrd2006k} and \texttt{Gurobi 12.0} \cite{gurobi}. We compared our approach against three baselines: (1) solving formulation~\eqref{eq:miqcp-NE} directly using Gurobi, which in version 12 applies similar transformations as those described by~\citet{ganzfried2024fast}; (2) the CMR algorithm proposed and open-sourced by~\citet{zhang2023computing}; and (3) the incomplete ADIDAS algorithm introduced by~\citet{gemp2022adidas} and implemented in the OpenSpiel library~\cite{LanctotEtAl2019OpenSpiel} which neither of the complete methods previously compared to. Since the \eqref{eq:miqcp-NE} formulation solved with Gurobi and the CMR algorithm have already been shown to outperform classical methods such as those implemented in the Gambit library~\cite{savani2025gambit} and the exclusion method~\cite{berg2017exclusion}, we focused our comparisons on these more competitive baselines.

The computational results are shown in the left-hand plots of Figure~\ref{fig:res}, where the vertical axis is presented on a logarithmic scale. All computations were terminated after one hour if not completed, and such cases are indicated with this time limit in the plots. We refer to our spatial branch-and-bound algorithm as SBnB, and to its early-terminated variant, which stops once the approximation quality of ADIDAS is approximately matched or exceeded, as SBnB-e. Out of the 96 tested instances, formulation~\eqref{eq:miqcp-NE} timed out in 77 cases, and CMR in 76 cases, solving 19.8\% and 20.8\% of the testing instances, respectively. In contrast, both variants of our algorithm, as well as ADIDAS, successfully computed an equilibrium for all instances. With the exception of several instances in the RoadGraph class, where CMR occasionally produced better results, SBnB and its early-terminated variant consistently outperformed the other complete methods. Notably, SBnB-e achieved not only faster runtimes than the incomplete ADIDAS algorithm on smaller instances but also comparable runtimes on larger games with 6 players and 5 or 6 actions.

The right-hand plots in Figure~\ref{fig:res} show the equilibrium approximation quality, measured as the $\epsilon$ of the computed solutions -- a minimal value satisfying Equations~\eqref{eps-nash}.\footnote{Note that the $\epsilon$ reported by default in the OpenSpiel implementation of ADIDAS corresponds to the average player-wise $\epsilon$, rather than the maximum deviation from best response.} For Gurobi-based methods (formulation~\eqref{eq:miqcp-NE} and CMR) the resulting $\epsilon$ values are on the order of Gurobi's numerical precision, typically around $10^{-5}$ to $10^{-6}$. SBnB consistently produces highly accurate solutions, in some cases achieving several orders of magnitude greater precision than alternative methods. Interestingly, SBnB-e performs comparably to SBnB on smaller instances, offering significant improvements over ADIDAS in both accuracy and runtime, and matches its solution quality on larger games.

We also considered the observation by~\citet{sandholm2005mixed} that adding a welfare-maximizing or support-minimizing objective can accelerate solver-based methods. While we observed a speedup on random games, it had no consistent effect on the instances reported here. 

%% file: figs/runtimes_complete.tex
\begin{figure*}[h!]
\centering
\def\bracescoord{0.008}
\def\labelcoord{0.005}
\def\bracescoordapp{0.000000000011}
\def\labelcoordapp{0.0000000000044}
\def\xscale{1.1}
\def\yscale{0.93}
\def\btwspace{.3cm}
\begin{tikzpicture}
\begin{semilogyaxis}[
    title={Runtimes for CompleteGraph},
    ylabel={Time (s)},
    width=\fpeval{\xscale*\axisdefaultwidth},
    height=\fpeval{\yscale*\axisdefaultheight},
    xtick={1,2,3,4,5,6,7,8,9,10,11,12,13,14,15,16,17,18,19,20,21,22,23,24,25,26,27,28,29,30},
    xticklabels={2,3,4,6,8,10,3,6,7,1,2,3,4,10,3,4,5,7,8,1,6,7,3,4,6,8,4,5,6,7},
    legend pos=south east,
    legend style={font=\footnotesize},
    legend columns=1,
    ymin=0.02,ymax=10000,
    ymajorgrids=true,
    grid style=dashed,
    xticklabel style={font=\tiny},
    clip=false,  
    enlargelimits=false,
    xmin=0,xmax=31
]

\addplot[blue, thick, dashed] coordinates {(1,6) (2,3147) (3,1) (4,364) (5,9) (6,1) (7,3600) (8,3600) (9,3600) (10,3600) (11,3600) (12,3600) (13,3600) (14,3600) (15,3600) (16,3600) (17,3601) (18,3600) (19,2) (20,3600) (21,3600) (22,3600) (23,3600) (24,3600) (25,3600) (26,3600) (27,3600) (28,3600) (29,3601) (30,3600)};
\addlegendentry{\eqref{eq:miqcp-NE}}


\addplot[blue, thick] coordinates {(1,5) (2,3600) (3,15) (4,209) (5,1) (6,13) (7,3600) (8,3600) (9,3600) (10,3600) (11,3600) (12,3600) (13,3600) (14,3600) (15,3600) (16,41) (17,3601) (18,3600) (19,3600) (20,3600) (21,3600) (22,3600) (23,3600) (24,3600) (25,3600) (26,3601) (27,3600) (28,3600) (29,3600) (30,3600)};
\addlegendentry{CMR}

\addplot[red, thick, dashed] coordinates {(1,6) (2,7) (3,6) (4,6) (5,6) (6,6) (7,7) (8,8) (9,8) (10,8) (11,9) (12,8) (13,8) (14,8) (15,8) (16,7) (17,7) (18,8) (19,7) (20,9) (21,8) (22,9) (23,9) (24,9) (25,10) (26,9) (27,12) (28,11) (29,12) (30,11)};
\addlegendentry{ADIDAS}

\addplot[black, thick,dashed] coordinates {(1,0.133442) (2,0.178695) (3,0.085566) (4,0.174809) (5,0.240685) (6,3.68829) (7,0.961914) (8,5.39915) (9,1208.83) (10,13.2187) (11,15.2468) (12,20.3398) (13,15.6713) (14,17.5101) (15,1.05695) (16,0.658083) (17,1.09595) (18,0.819205) (19,2.55535) (20,8.97233) (21,17.3329) (22,6.91171) (23,22.8083) (24,27.9547) (25,79.9979) (26,28.1257) (27,233.493) (28,178.892) (29,171.396) (30,287.503)};
\addlegendentry{SBnB}

\addplot[black, thick] coordinates {(1,0.095257) (2,0.08298) (3,0.059692) (4,0.051093) (5,0.057052) (6,3.83568) (7,0.586193) (8,1.70498) (9,2.85846) (10,5.2835) (11,5.09819) (12,5.41909) (13,5.27105) (14,5.42846) (15,0.780195) (16,0.411544) (17,0.323369) (18,0.495506) (19,0.711678) (20,4.68208) (21,4.91963) (22,4.92321) (23,14.5619) (24,15.1541) (25,15.5261) (26,15.1674) (27,16.2126) (28,16.1643) (29,16.2533) (30,16.4048)};
\addlegendentry{SBnB-e}


\draw[decorate, decoration={brace, mirror, amplitude=5pt}, thick]
  (axis  cs:0.7, \bracescoord) -- (axis  cs:6.3, \bracescoord);

\draw[decorate, decoration={brace, mirror, amplitude=5pt}, thick]
  (axis  cs:6.7, \bracescoord) -- (axis  cs:7.3, \bracescoord);

\draw[decorate, decoration={brace, mirror, amplitude=5pt}, thick]
  (axis  cs:7.7, \bracescoord) -- (axis  cs:9.3, \bracescoord);

\draw[decorate, decoration={brace, mirror, amplitude=5pt}, thick]
  (axis  cs:9.7, \bracescoord) -- (axis  cs:14.3, \bracescoord);

\draw[decorate, decoration={brace, mirror, amplitude=5pt}, thick]
  (axis  cs:14.7, \bracescoord) -- (axis  cs:19.3, \bracescoord);

\draw[decorate, decoration={brace, mirror, amplitude=5pt}, thick]
  (axis  cs:19.7, \bracescoord) -- (axis  cs:22.3, \bracescoord);

\draw[decorate, decoration={brace, mirror, amplitude=5pt}, thick]
  (axis  cs:22.7, \bracescoord) -- (axis  cs:26.3, \bracescoord);

\draw[decorate, decoration={brace, mirror, amplitude=5pt}, thick]
  (axis  cs:26.7, \bracescoord) -- (axis  cs:30.3, \bracescoord);

\node[anchor=north, font=\footnotesize] at (axis cs:3.5, \labelcoord) {5/3};
\node[anchor=north, font=\footnotesize] at (axis cs:6.6, \labelcoord) {5/4};
\node[anchor=north, font=\footnotesize] at (axis cs:9, \labelcoord) {5/5};
\node[anchor=north, font=\footnotesize] at (axis cs:12, \labelcoord) {5/6};
\node[anchor=north, font=\footnotesize] at (axis cs:17, \labelcoord) {6/3};
\node[anchor=north, font=\footnotesize] at (axis cs:21, \labelcoord) {6/4};
\node[anchor=north, font=\footnotesize] at (axis cs:24.5, \labelcoord) {6/5};
\node[anchor=north, font=\footnotesize] at (axis cs:28.5, \labelcoord) {6/6};

\end{semilogyaxis}
\end{tikzpicture}
\hfill
\begin{tikzpicture}
\begin{semilogyaxis}[
    title={Approximation quality for CompleteGraph},
    ylabel={Nash $\epsilon$},
    height=\fpeval{\yscale*\axisdefaultheight},
    xtick={1,2,3,4,5,6,7,8,9,10,11,12,13,14,15,16,17,18,19,20,21,22,23,24,25,26,27,28,29,30},
    xticklabels={2,3,4,6,8,10,3,6,7,1,2,3,4,10,3,4,5,7,8,1,6,7,3,4,6,8,4,5,6,7},
    legend pos=south east,
    legend style={font=\small},
    legend columns=1,
    ymajorgrids=true,
    grid style=dashed,
    xticklabel style={font=\tiny},
    clip=false,  
    enlargelimits=false,
    xmin=0,xmax=31,
    ymin=5e-11,ymax=7e-1
]
\addplot[red, thick, dashed] coordinates {(1,3.133649e-02) (2,5.985827e-02) (3,3.956318e-02) (4,7.095255e-02) (5,4.029424e-02) (6,5.375118e-02) (7,3.192626e-02) (8,1.800156e-02) (9,2.130996e-02) (10,1.772353e-02) (11,1.562405e-02) (12,1.161808e-02) (13,1.584674e-02) (14,2.754504e-02) (15,3.355921e-02) (16,3.250440e-02) (17,2.205725e-02) (18,2.432907e-02) (19,2.893233e-02) (20,2.354401e-02) (21,1.158106e-02) (22,1.587433e-02) (23,6.567992e-03) (24,9.285572e-03) (25,1.233096e-02) (26,7.651204e-03) (27,5.798986e-03) (28,5.330395e-03) (29,4.578139e-03) (30,8.258732e-03)};
\addplot[black, thick, dashed] coordinates {(1,2.22988e-8) (2,1.2538e-8) (3,8.73177e-9) (4,3.6073e-9) (5,7.33029e-10) (6,9.48565e-10) (7,9.86588e-11) (8,5.51331e-9) (9,1.24555e-8) (10,3.49879e-9) (11,2.03616e-7) (12,3.68019e-5) (13,3.53679e-9) (14,7.93838e-7) (15,1.78954e-8) (16,7.99541e-10) (17,1.37133e-9) (18,4.99831e-9) (19,8.18463e-8) (20,2.0783e-8) (21,7.82079e-10) (22,8.31389e-10) (23,6.41153e-10) (24,2.43449e-9) (25,1.5777e-8) (26,5.69917e-9) (27,4.6971e-7) (28,1.21255e-8) (29,2.35685e-9) (30,1.12701e-6)};
\addplot[black, thick] coordinates {(1,1.0952e-7) (2,1.11507e-9) (3,3.39716e-8) (4,9.31525e-8) (5,0.00173507) (6,6.73909e-8) (7,0.000153981) (8,0.000300697) (9,0.000240381) (10,9.08982e-9) (11,2.03616e-7) (12,3.68019e-5) (13,0.000228559) (14,7.93838e-7) (15,3.22914e-8) (16,7.15837e-8) (17,2.2824e-8) (18,1.4436e-8) (19,8.18463e-8) (20,0.000284234) (21,7.74253e-9) (22,0.000133086) (23,1.55722e-6) (24,0.000199649) (25,0.00012541) (26,4.0051e-5) (27,0.00939041) (28,0.00939918) (29,0.0088543) (30,0.010354)};

\draw[decorate, decoration={brace, mirror, amplitude=5pt}, thick]
  (axis  cs:0.7, \bracescoordapp) -- (axis  cs:6.3, \bracescoordapp);

\draw[decorate, decoration={brace, mirror, amplitude=5pt}, thick]
  (axis  cs:6.7, \bracescoordapp) -- (axis  cs:7.3, \bracescoordapp);

\draw[decorate, decoration={brace, mirror, amplitude=5pt}, thick]
  (axis  cs:7.7, \bracescoordapp) -- (axis  cs:9.3, \bracescoordapp);

\draw[decorate, decoration={brace, mirror, amplitude=5pt}, thick]
  (axis  cs:9.7, \bracescoordapp) -- (axis  cs:14.3, \bracescoordapp);

\draw[decorate, decoration={brace, mirror, amplitude=5pt}, thick]
  (axis  cs:14.7, \bracescoordapp) -- (axis  cs:19.3, \bracescoordapp);

\draw[decorate, decoration={brace, mirror, amplitude=5pt}, thick]
  (axis  cs:19.7, \bracescoordapp) -- (axis  cs:22.3, \bracescoordapp);

\draw[decorate, decoration={brace, mirror, amplitude=5pt}, thick]
  (axis  cs:22.7, \bracescoordapp) -- (axis  cs:26.3, \bracescoordapp);

\draw[decorate, decoration={brace, mirror, amplitude=5pt}, thick]
  (axis  cs:26.7, \bracescoordapp) -- (axis  cs:30.3, \bracescoordapp);

\node[anchor=north, font=\footnotesize] at (axis cs:3.5, \labelcoordapp) {5/3};
\node[anchor=north, font=\footnotesize] at (axis cs:6.6, \labelcoordapp) {5/4};
\node[anchor=north, font=\footnotesize] at (axis cs:9, \labelcoordapp) {5/5};
\node[anchor=north, font=\footnotesize] at (axis cs:12, \labelcoordapp) {5/6};
\node[anchor=north, font=\footnotesize] at (axis cs:17, \labelcoordapp) {6/3};
\node[anchor=north, font=\footnotesize] at (axis cs:21, \labelcoordapp) {6/4};
\node[anchor=north, font=\footnotesize] at (axis cs:24.5, \labelcoordapp) {6/5};
\node[anchor=north, font=\footnotesize] at (axis cs:28.5, \labelcoordapp) {6/6};
\end{semilogyaxis}
\end{tikzpicture}

\vspace{\btwspace}
\begin{tikzpicture}
\begin{semilogyaxis}[
    title={Runtimes for RoadGraph},
    ylabel={Time (s)},
    width=\fpeval{\xscale*\axisdefaultwidth},
    height=\fpeval{\yscale*\axisdefaultheight},
    xtick={1,2,3,4,5,6,7,8,9,10,11,12,13,14,15,16,17,18,19,20,21,22,23,24,25,26,27,28,29,30},
    xticklabels={1,2,4,5,10,2,8,1,3,5,7,3,5,4,5,6,1,3,9,10,2,3,4,8,3,4,5,6,8,9},
    legend pos=south east,
    legend style={font=\small},
    legend columns=1,
    ymin=0.02,ymax=10000,
    ymajorgrids=true,
    grid style=dashed,
    xticklabel style={font=\tiny},
    clip=false,  
    enlargelimits=false,
    xmin=0,xmax=31
]

\addplot[blue, thick, dashed] coordinates {(1,13) (2,1) (3,1) (4,1) (5,1) (6,3600) (7,1) (8,3600) (9,3600) (10,3600) (11,3600) (12,3600) (13,3600) (14,2) (15,3600) (16,3600) (17,3600) (18,3600) (19,3600) (20,3) (21,3600) (22,3600) (23,3600) (24,3600) (25,3600) (26,3600) (27,3601) (28,3600) (29,3600) (30,3601)};
\addlegendentry{\eqref{eq:miqcp-NE}}


\addplot[blue, thick] coordinates {(1,1) (2,1) (3,0) (4,1) (5,0) (6,1556) (7,2) (8,3600) (9,3600) (10,1926) (11,1324) (12,3600) (13,3600) (14,11) (15,2) (16,3) (17,3600) (18,3600) (19,3600) (20,3075) (21,3600) (22,3600) (23,3600) (24,3600) (25,3601) (26,3600) (27,3600) (28,3601) (29,59) (30,3600)};
\addlegendentry{CMR}

\addplot[red, thick, dashed] coordinates {(1,6) (2,7) (3,6) (4,6) (5,6) (6,7) (7,7) (8,8) (9,7) (10,8) (11,8) (12,8) (13,8) (14,8) (15,7) (16,8) (17,8) (18,9) (19,8) (20,8) (21,9) (22,10) (23,9) (24,10) (25,11) (26,12) (27,11) (28,11) (29,12) (30,11)};
\addlegendentry{ADIDAS}
\addplot[black, thick, dashed] coordinates {(1,0.249063) (2,0.070656) (3,0.064888) (4,0.157574) (5,0.21233) (6,1.72074) (7,0.577108) (8,4.97368) (9,4.54667) (10,4.97981) (11,1216.14) (12,16.6371) (13,22.275) (14,0.455278) (15,0.499689) (16,584.826) (17,10.8316) (18,14.7011) (19,11.6209) (20,1215.87) (21,1340.94) (22,20.7) (23,53.3175) (24,61.6802) (25,118.295) (26,1704.08) (27,112.597) (28,170.386) (29,193.043) (30,235.012)};
\addlegendentry{SBnB}
\addplot[black, thick] coordinates {(1,0.050719) (2,0.049882) (3,0.040011) (4,0.059672) (5,0.080782) (6,0.460137) (7,0.365334) (8,2.20608) (9,1.95811) (10,1.4245) (11,1.82755) (12,9.49069) (13,4.91715) (14,0.271996) (15,0.306132) (16,0.340008) (17,5.17166) (18,2.83717) (19,4.18765) (20,5.58033) (21,15.0168) (22,12.8595) (23,15.1945) (24,15.0939) (25,16.203) (26,16.136) (27,16.2344) (28,16.2592) (29,16.2465) (30,16.2314)};
\addlegendentry{SBnB-e}

\draw[decorate, decoration={brace, mirror, amplitude=5pt}, thick]
  (axis  cs:0.7, \bracescoord) -- (axis  cs:5.3, \bracescoord);

\draw[decorate, decoration={brace, mirror, amplitude=5pt}, thick]
  (axis  cs:5.7, \bracescoord) -- (axis  cs:7.3, \bracescoord);

\draw[decorate, decoration={brace, mirror, amplitude=5pt}, thick]
  (axis  cs:7.7, \bracescoord) -- (axis  cs:11.3, \bracescoord);

\draw[decorate, decoration={brace, mirror, amplitude=5pt}, thick]
  (axis  cs:11.7, \bracescoord) -- (axis  cs:13.3, \bracescoord);

\draw[decorate, decoration={brace, mirror, amplitude=5pt}, thick]
  (axis  cs:13.7, \bracescoord) -- (axis  cs:16.3, \bracescoord);

\draw[decorate, decoration={brace, mirror, amplitude=5pt}, thick]
  (axis  cs:16.7, \bracescoord) -- (axis  cs:20.3, \bracescoord);

\draw[decorate, decoration={brace, mirror, amplitude=5pt}, thick]
  (axis  cs:20.7, \bracescoord) -- (axis  cs:24.3, \bracescoord);

\draw[decorate, decoration={brace, mirror, amplitude=5pt}, thick]
  (axis  cs:24.7, \bracescoord) -- (axis  cs:30.3, \bracescoord);

\node[anchor=north, font=\footnotesize] at (axis cs:3, \labelcoord) {5/3};
\node[anchor=north, font=\footnotesize] at (axis cs:6.5, \labelcoord) {5/4};
\node[anchor=north, font=\footnotesize] at (axis cs:9.5, \labelcoord) {5/5};
\node[anchor=north, font=\footnotesize] at (axis cs:12.5, \labelcoord) {5/6};
\node[anchor=north, font=\footnotesize] at (axis cs:15, \labelcoord) {6/3};
\node[anchor=north, font=\footnotesize] at (axis cs:18.5, \labelcoord) {6/4};
\node[anchor=north, font=\footnotesize] at (axis cs:22.5, \labelcoord) {6/5};
\node[anchor=north, font=\footnotesize] at (axis cs:27.5, \labelcoord) {6/6};

\end{semilogyaxis}
\end{tikzpicture}
\hfill
\begin{tikzpicture}
\begin{semilogyaxis}[
    title={Approximation quality for RoadGraph},
    ylabel={Nash $\epsilon$},
    height=\fpeval{\yscale*\axisdefaultheight},
    xtick={1,2,3,4,5,6,7,8,9,10,11,12,13,14,15,16,17,18,19,20,21,22,23,24,25,26,27,28,29,30},
    xticklabels={1,2,4,5,10,2,8,1,3,5,7,3,5,4,5,6,1,3,9,10,2,3,4,8,3,4,5,6,8,9},
    legend pos=south east,
    legend style={font=\small},
    legend columns=1,
    ymajorgrids=true,
    grid style=dashed,
    xticklabel style={font=\tiny},
    clip=false,  
    enlargelimits=false,
    xmin=0,xmax=31,
    ymin=5e-11,ymax=7e-1
]
\addplot[red, thick, dashed] coordinates {(1,9.691819e-02) (2,1.315890e-01) (3,2.801595e-01) (4,1.341540e-01) (5,1.558555e-01) (6,1.176729e-01) (7,2.192296e-01) (8,2.100876e-01) (9,8.480640e-02) (10,1.675977e-01) (11,2.577948e-01) (12,1.148518e-01) (13,7.726770e-02) (14,1.691721e-01) (15,1.941211e-01) (16,1.682660e-01) (17,9.822253e-02) (18,1.552035e-01) (19,6.648668e-02) (20,1.197585e-01) (21,1.318729e-01) (22,9.556964e-02) (23,7.158419e-02) (24,9.611138e-02) (25,7.334500e-02) (26,6.226135e-02) (27,9.765877e-02) (28,6.914702e-02) (29,9.544834e-02) (30,6.756572e-02)};
\addplot[black, thick, dashed] coordinates {(1,1.16989e-8) (2,4.23495e-10) (3,1.77533e-9) (4,4.74703e-10) (5,2.00494e-9) (6,8.89307e-5) (7,1.09149e-7) (8,1.50587e-8) (9,4.95102e-9) (10,3.61147e-9) (11,0.0140142) (12,8.87736e-10) (13,1.48973e-10) (14,2.31075e-10) (15,2.03285e-10) (16,7.5236e-8) (17,8.72347e-9) (18,2.12117e-10) (19,6.98129e-10) (20,0.000883231) (21,0.000697142) (22,2.79097e-9) (23,1.856e-7) (24,2.42392e-8) (25,3.43511e-8) (26,0.00235834) (27,1.15263e-7) (28,9.51647e-10) (29,6.94748e-9) (30,5.0681e-8)};
\addplot[black, thick] coordinates {(1,1.16989e-8) (2,5.34093e-10) (3,1.91907e-7) (4,0.00526779) (5,6.70534e-7) (6,8.89307e-5) (7,1.97097e-7) (8,0.00664386) (9,0.000595982) (10,4.35516e-7) (11,0.00160268) (12,2.28089e-6) (13,1.99285e-8) (14,1.76782e-6) (15,0.000757163) (16,0.0102146) (17,0.0020351) (18,4.80946e-8) (19,0.00157044) (20,0.000883231) (21,0.00803785) (22,2.5966e-6) (23,0.00444961) (24,0.0125356) (25,-12.2311) (26,0.104779) (27,0.0847244) (28,0.09825) (29,0.123478) (30,0.0886908)};

\draw[decorate, decoration={brace, mirror, amplitude=5pt}, thick]
  (axis  cs:0.7, \bracescoordapp) -- (axis  cs:5.3, \bracescoordapp);

\draw[decorate, decoration={brace, mirror, amplitude=5pt}, thick]
  (axis  cs:5.7, \bracescoordapp) -- (axis  cs:7.3, \bracescoordapp);

\draw[decorate, decoration={brace, mirror, amplitude=5pt}, thick]
  (axis  cs:7.7, \bracescoordapp) -- (axis  cs:11.3, \bracescoordapp);

\draw[decorate, decoration={brace, mirror, amplitude=5pt}, thick]
  (axis  cs:11.7, \bracescoordapp) -- (axis  cs:13.3, \bracescoordapp);

\draw[decorate, decoration={brace, mirror, amplitude=5pt}, thick]
  (axis  cs:13.7, \bracescoordapp) -- (axis  cs:16.3, \bracescoordapp);

\draw[decorate, decoration={brace, mirror, amplitude=5pt}, thick]
  (axis  cs:16.7, \bracescoordapp) -- (axis  cs:20.3, \bracescoordapp);

\draw[decorate, decoration={brace, mirror, amplitude=5pt}, thick]
  (axis  cs:20.7, \bracescoordapp) -- (axis  cs:24.3, \bracescoordapp);

\draw[decorate, decoration={brace, mirror, amplitude=5pt}, thick]
  (axis  cs:24.7, \bracescoordapp) -- (axis  cs:30.3, \bracescoordapp);

\node[anchor=north, font=\footnotesize] at (axis cs:3, \labelcoordapp) {5/3};
\node[anchor=north, font=\footnotesize] at (axis cs:6.5, \labelcoordapp) {5/4};
\node[anchor=north, font=\footnotesize] at (axis cs:9.5, \labelcoordapp) {5/5};
\node[anchor=north, font=\footnotesize] at (axis cs:12.5, \labelcoordapp) {5/6};
\node[anchor=north, font=\footnotesize] at (axis cs:15, \labelcoordapp) {6/3};
\node[anchor=north, font=\footnotesize] at (axis cs:18.5, \labelcoordapp) {6/4};
\node[anchor=north, font=\footnotesize] at (axis cs:22.5, \labelcoordapp) {6/5};
\node[anchor=north, font=\footnotesize] at (axis cs:27.5, \labelcoordapp) {6/6};

\end{semilogyaxis}
\end{tikzpicture}

\vspace{\btwspace}
\begin{tikzpicture}
\begin{semilogyaxis}[
    title={Runtimes for SmallWorldGraph},
    width=\fpeval{\xscale*\axisdefaultwidth},
    height=\fpeval{\yscale*\axisdefaultheight},
    ylabel={Time (s)},
    xtick={1,2,3,4,5,6,7,8,9,10,11,12,13,14,15,16,17,18,19,20,21,22,23,24,25,26,27,28,29,30,31,32,33,34,35,36},
    xticklabels={6,14,21,24,26,27,28,29,30,5,11,12,17,20,21,22,27,1,8,9,12,14,15,21,22,25,1,2,3,7,12,17,20,22,24,28},
    legend pos=south east,
    legend style={font=\small},
    legend columns=1,
    ymin=0.02,ymax=10000,
    ymajorgrids=true,
    grid style=dashed,
    xticklabel style={font=\fontsize{5pt}{6pt}\selectfont},
    clip=false,  
    enlargelimits=false,
    xmin=0,xmax=37
]

\addplot[blue, thick, dashed] coordinates {(1,1) (2,3600) (3,2262) (4,1) (5,3600) (6,3600) (7,3600) (8,2) (9,3600) (10,3600) (11,3600) (12,3600) (13,3600) (14,3600) (15,3600) (16,3600) (17,3600) (18,3600) (19,3600) (20,3600) (21,3600) (22,3600) (23,3600) (24,3600) (25,3600) (26,3600) (27,3600) (28,3601) (29,3600) (30,3600) (31,3600) (32,3601) (33,3600) (34,3600) (35,3601) (36,3601)};
\addlegendentry{\eqref{eq:miqcp-NE}}


\addplot[blue, thick] coordinates {(1,3600) (2,3600) (3,3600) (4,3600) (5,3601) (6,3600) (7,3600) (8,3600) (9,3600) (10,3601) (11,3600) (12,3600) (13,3600) (14,3601) (15,3600) (16,3600) (17,3600) (18,3600) (19,3601) (20,3600) (21,3600) (22,3601) (23,3600) (24,3600) (25,3601) (26,3600) (27,3601) (28,3600) (29,3601) (30,3600) (31,3600) (32,3600) (33,3601) (34,3600) (35,3600) (36,3600)};
\addlegendentry{CMR}

\addplot[red, thick, dashed] coordinates {(1,8) (2,7) (3,8) (4,7) (5,8) (6,7) (7,8) (8,7) (9,7) (10,9) (11,8) (12,8) (13,9) (14,8) (15,8) (16,9) (17,8) (18,9) (19,10) (20,9) (21,10) (22,9) (23,10) (24,9) (25,10) (26,9) (27,12) (28,11) (29,11) (30,11) (31,11) (32,11) (33,12) (34,11) (35,11) (36,12)};
\addlegendentry{ADIDAS}

\addplot[black, thick, dashed] coordinates {(1,0.389382) (2,0.812209) (3,0.663268) (4,3.13335) (5,0.651624) (6,0.583531) (7,0.876516) (8,0.97428) (9,1.56981) (10,4.31233) (11,3.99826) (12,9.64406) (13,10.1105) (14,12.8035) (15,7.658) (16,12.465) (17,15.7251) (18,83.9992) (19,25.9613) (20,32.0505) (21,64.6879) (22,51.6826) (23,24.0361) (24,81.3601) (25,1332.78) (26,26.7835) (27,84.9778) (28,265.183) (29,105.859) (30,176.64) (31,232.291) (32,97.5154) (33,98.8438) (34,104.734) (35,110.998) (36,93.3736)};
\addlegendentry{SBnB}

\addplot[black, thick] coordinates {(1,0.207033) (2,0.542504) (3,0.308074) (4,0.320789) (5,0.34236) (6,0.4107) (7,0.457214) (8,0.2405) (9,0.327369) (10,3.11324) (11,2.729) (12,3.41392) (13,4.57462) (14,4.67692) (15,3.3975) (16,3.11839) (17,4.40487) (18,12.8364) (19,14.6172) (20,15.1378) (21,15.0517) (22,15.3315) (23,13.9692) (24,15.0882) (25,15.1364) (26,14.002) (27,16.2105) (28,16.2406) (29,16.519) (30,16.2509) (31,16.2487) (32,16.2171) (33,16.2081) (34,16.1659) (35,16.2005) (36,16.1254)};
\addlegendentry{SBnB-e}

\draw[decorate, decoration={brace, mirror, amplitude=5pt}, thick]
  (axis  cs:0.7, \bracescoord) -- (axis  cs:9.3, \bracescoord);

\draw[decorate, decoration={brace, mirror, amplitude=5pt}, thick]
  (axis  cs:9.7, \bracescoord) -- (axis  cs:17.3, \bracescoord);

\draw[decorate, decoration={brace, mirror, amplitude=5pt}, thick]
  (axis  cs:17.7, \bracescoord) -- (axis  cs:26.3, \bracescoord);

\draw[decorate, decoration={brace, mirror, amplitude=5pt}, thick]
  (axis  cs:26.7, \bracescoord) -- (axis  cs:36.3, \bracescoord);

\node[anchor=north, font=\footnotesize] at (axis cs:5, \labelcoord) {6/3};
\node[anchor=north, font=\footnotesize] at (axis cs:13.5, \labelcoord) {6/4};
\node[anchor=north, font=\footnotesize] at (axis cs:22, \labelcoord) {6/5};
\node[anchor=north, font=\footnotesize] at (axis cs:31.5, \labelcoord) {6/6};

\end{semilogyaxis}
\end{tikzpicture}
\hfill
\begin{tikzpicture}
\begin{semilogyaxis}[
    title={Approximation quality for SmallWorldGraph},
    ylabel={Nash $\epsilon$},
    height=\fpeval{\yscale*\axisdefaultheight},
    xtick={1,2,3,4,5,6,7,8,9,10,11,12,13,14,15,16,17,18,19,20,21,22,23,24,25,26,27,28,29,30,31,32,33,34,35,36},
    xticklabels={6,14,21,24,26,27,28,29,30,5,11,12,17,20,21,22,27,1,8,9,12,14,15,21,22,25,1,2,3,7,12,17,20,22,24,28},
    legend pos=south east,
    legend style={font=\small},
    legend columns=1,
    ymajorgrids=true,
    grid style=dashed,
    xticklabel style={font=\fontsize{4pt}{5pt}\selectfont},
    clip=false,  
    enlargelimits=false,
    xmin=0,xmax=37,
    ymin=5e-11,ymax=7e-1
]
\addplot[red, thick, dashed] coordinates {(1,7.216928e-02) (2,7.912203e-02) (3,1.122020e-01) (4,6.994430e-02) (5,7.359067e-02) (6,6.776958e-02) (7,3.730288e-02) (8,8.621069e-02) (9,6.301901e-02) (10,9.537178e-02) (11,1.924232e-02) (12,3.946829e-02) (13,4.008375e-02) (14,6.941653e-02) (15,6.262732e-02) (16,2.851353e-02) (17,4.469049e-02) (18,2.098038e-02) (19,2.353227e-02) (20,2.667401e-02) (21,2.192388e-02) (22,3.869713e-02) (23,3.714909e-02) (24,7.711832e-02) (25,1.885536e-02) (26,1.921273e-02) (27,3.186175e-02) (28,5.510527e-02) (29,2.959710e-02) (30,4.644845e-02) (31,3.375689e-02) (32,2.405738e-02) (33,1.452646e-02) (34,2.995259e-02) (35,1.862746e-02) (36,3.051766e-02)};
\addplot[black, thick, dashed] coordinates {(1,7.59712e-8) (2,1.06317e-7) (3,2.14375e-9) (4,3.94381e-7) (5,2.00683e-9) (6,3.28909e-9) (7,2.82495e-7) (8,7.70712e-11) (9,8.92906e-9) (10,4.58084e-9) (11,2.04679e-8) (12,5.76194e-7) (13,3.60952e-9) (14,5.07446e-9) (15,5.61213e-10) (16,2.9798e-10) (17,8.99691e-5) (18,1.63594e-7) (19,1.04945e-10) (20,1.36618e-7) (21,8.46863e-11) (22,2.90866e-9) (23,3.91586e-8) (24,2.68274e-9) (25,0.000199919) (26,3.10325e-8) (27,2.03235e-10) (28,1.36129e-10) (29,8.94896e-8) (30,8.05232e-9) (31,9.77576e-7) (32,2.53937e-9) (33,4.13003e-9) (34,4.64543e-10) (35,1.7021e-8) (36,8.41133e-9)};
\addplot[black, thick] coordinates {(1,2.75631e-8) (2,0.00111919) (3,1.18088e-7) (4,0.00190872) (5,0.00094469) (6,7.87696e-5) (7,7.24122e-10) (8,9.32615e-8) (9,1.15537e-6) (10,3.20764e-8) (11,1.07026e-8) (12,0.000837698) (13,0.000101235) (14,3.09257e-8) (15,5.46604e-7) (16,5.38577e-8) (17,8.99691e-5) (18,1.63594e-7) (19,7.8561e-8) (20,0.000186396) (21,4.75701e-5) (22,0.00213532) (23,3.43972e-6) (24,0.000431454) (25,0.00029042) (26,0.00015623) (27,0.0253009) (28,0.130714) (29,0.0384712) (30,0.0400614) (31,0.0520221) (32,0.0232766) (33,0.0146924) (34,0.0248141) (35,0.0292429) (36,0.0368252)};

\draw[decorate, decoration={brace, mirror, amplitude=5pt}, thick]
  (axis  cs:0.7, \bracescoordapp) -- (axis  cs:9.3, \bracescoordapp);

\draw[decorate, decoration={brace, mirror, amplitude=5pt}, thick]
  (axis  cs:9.7, \bracescoordapp) -- (axis  cs:17.3, \bracescoordapp);

\draw[decorate, decoration={brace, mirror, amplitude=5pt}, thick]
  (axis  cs:17.7, \bracescoordapp) -- (axis  cs:26.3, \bracescoordapp);

\draw[decorate, decoration={brace, mirror, amplitude=5pt}, thick]
  (axis  cs:26.7, \bracescoordapp) -- (axis  cs:36.3, \bracescoordapp);

\node[anchor=north, font=\footnotesize] at (axis cs:5, \labelcoordapp) {6/3};
\node[anchor=north, font=\footnotesize] at (axis cs:13.5, \labelcoordapp) {6/4};
\node[anchor=north, font=\footnotesize] at (axis cs:22, \labelcoordapp) {6/5};
\node[anchor=north, font=\footnotesize] at (axis cs:31.5, \labelcoordapp) {6/6};
\end{semilogyaxis}
\end{tikzpicture}

\caption{Runtime and approximation quality results for Random Graphical Games on different underlying graphs. Each tick corresponds to a seed, the numbers below indicate the number of players / number of actions per player.}
\label{fig:res}
\end{figure*}
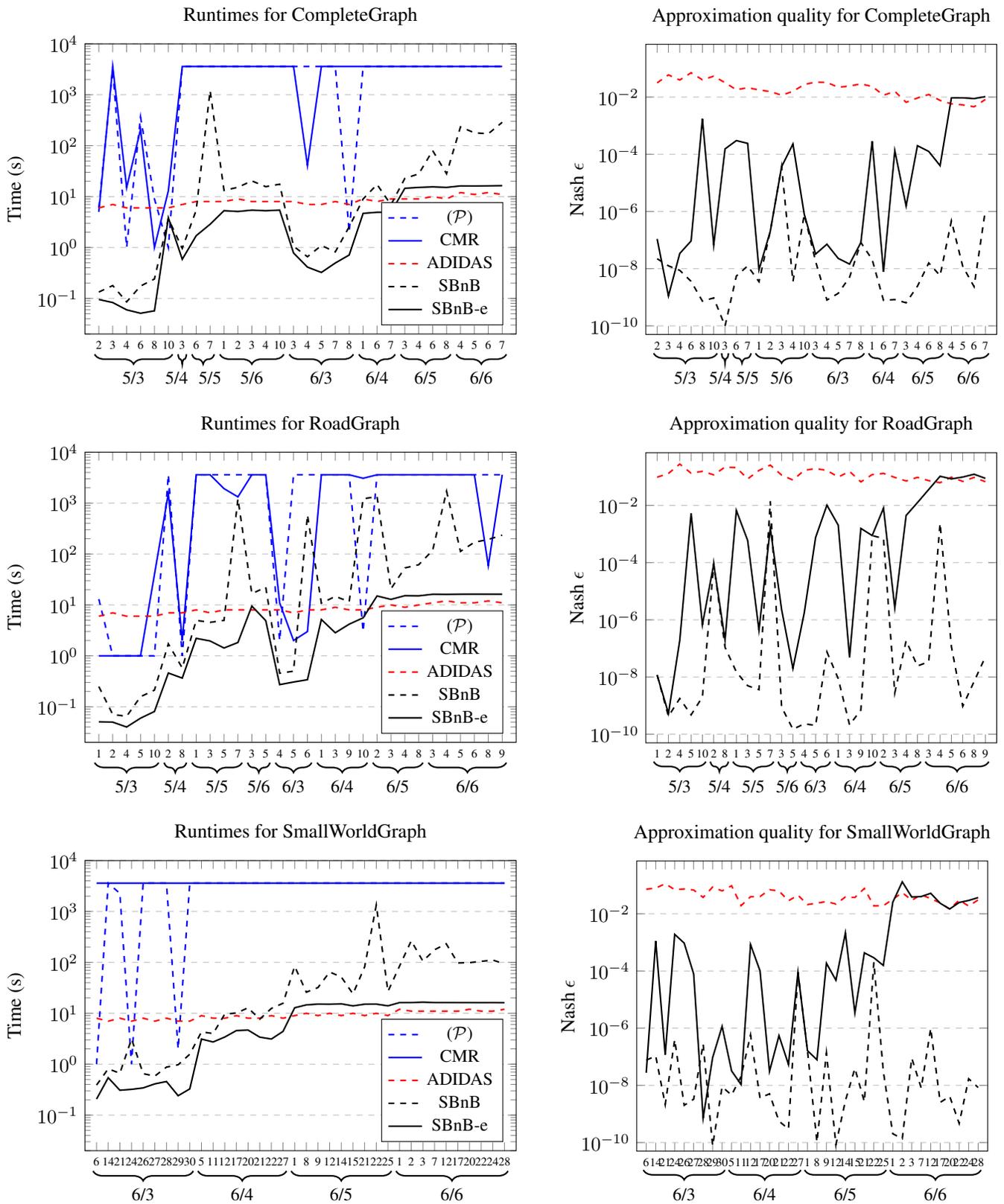

%% file: sections/conclusion.tex
\section{Conclusion}
We proposed solving Nash equilibria in multiplayer normal-form games via a polynomial complementarity problem, with a sound and complete spatial branch-and-bound algorithm. Our method offers a principled, scalable, and robust alternative to classical approaches, avoiding numerical instability while providing formal solution guarantees. Though challenges remain for very large games, our results indicate that general-purpose equilibrium computation is more tractable than previously thought, and our algorithm shows promise for integration into broader oracle-based solvers.

%% file: sections/checklist.tex
\section{Reproducibility Checklist}
\begin{enumerate}
    \item This paper:
    \begin{itemize}
        \item Includes a conceptual outline and/or pseudocode description of AI methods introduced (yes)
        \item Clearly delineates statements that are opinions, hypotheses, and speculations from objective facts and results (yes)
        \item Provides well-marked pedagogical references for less-familiar readers to gain background necessary to replicate the paper (yes)
    \end{itemize}

    \item Does this paper make theoretical contributions? (yes)

    If yes, please complete the list below.
    \begin{itemize}
        \item All assumptions and restrictions are stated clearly and formally. (yes)
        \item All novel claims are stated formally (e.g., in theorem statements). (yes)
        \item Proofs of all novel claims are included. (yes)
        \item Proof sketches or intuitions are given for complex and/or novel results. (yes)
        \item Appropriate citations to theoretical tools used are given. (yes)
        \item All theoretical claims are demonstrated empirically to hold. (yes)
        \item All experimental code used to eliminate or disprove claims is included. (NA)
    \end{itemize}

    \item Does this paper rely on one or more datasets? (yes)

    If yes, please complete the list below.
    \begin{itemize}
        \item A motivation is given for why the experiments are conducted on the selected datasets. (yes)
        \item All novel datasets introduced in this paper are included in a data appendix. (NA)
        \item All novel datasets introduced in this paper will be made publicly available upon publication with a license that allows free usage for research purposes. (NA)
        \item All datasets drawn from the existing literature are accompanied by appropriate citations. (yes)
        \item All datasets drawn from the existing literature are publicly available. (yes)
        \item All datasets that are not publicly available are described in detail, with an explanation of why publicly available alternatives are not scientifically sufficient. (/NA)
    \end{itemize}

    \item Does this paper include computational experiments? (yes)

    If yes, please complete the list below.
    \begin{itemize}
        \item States the number and range of values tried per (hyper-)parameter during development, along with the criterion for selecting final settings. (NA)
        \item Any code required for pre-processing data is included in the appendix. (yes)
        \item All source code for conducting and analyzing experiments is included in a code appendix. (yes)
        \item All source code will be made publicly available upon publication with a license that allows free research usage. (yes)
        \item All code implementing new methods includes comments referencing the paper sections each step originates from. (no)
        \item If algorithms depend on randomness, the method used for setting seeds is described sufficiently for replication. (yes)
        \item Specifies computing infrastructure (hardware/software), including GPU/CPU models, memory, OS, and library versions. (yes)
        \item Formally describes evaluation metrics and justifies their selection. (yes)
        \item States the number of algorithm runs used for each reported result. (NA)
        \item Analyzes experimental results beyond averages (e.g., includes variation, confidence intervals, etc.). (NA)
        \item Assesses performance changes using appropriate statistical tests (e.g., Wilcoxon signed-rank). (NA)
        \item Lists all final (hyper-)parameters used for each model/algorithm in the experiments. (NA)
    \end{itemize}
\end{enumerate}

%% file: sections/appendix.tex
\newpage
\section{Including the Objective}
As noted previously, we considered the observation by~\citet{sandholm2005mixed}, who found, somewhat counterintuitively, that incorporating a welfare-maximizing or support-minimizing objective into the formulation can substantially enhance the performance of solver-based algorithms in two-player general-sum games. While we observed runtime improvements on randomly generated instances, our findings on the GAMUT benchmark games were less conclusive. In particular, although the CMR algorithm exhibited a modest speedup with the additional objective on the RoadGraph and select smaller instances of SmallWorldGraph, its inclusion in the formulation~\eqref{eq:miqcp-NE} in some cases led to a noticeable slowdown. The results are included below, where \texttt{+o} denotes a formulation with a welfare-maximizing objective.
\input{figs/runtimes_appendix}

\newpage
\section{Performance on Random Games}
\input{figs/runtimes_random}
\input{figs/eps_random}

%% file: figs/runtimes_appendix.tex
\def\bracescoord{0.008}
\def\labelcoord{0.005}
\def\bracescoordapp{0.000000000011}
\def\labelcoordapp{0.0000000000044}
\def\xscale{1.3}
\def\yscale{1.0}
\def\btwspace{.3cm}
\begin{figure*}[h!]
\centering
\begin{tikzpicture}
\begin{semilogyaxis}[
    title={Runtimes for CompleteGraph},
    ylabel={Time (s)},
    width=\fpeval{\xscale*\axisdefaultwidth},
    height=\fpeval{\yscale*\axisdefaultheight},
    xtick={1,2,3,4,5,6,7,8,9,10,11,12,13,14,15,16,17,18,19,20,21,22,23,24,25,26,27,28,29,30},
    xticklabels={2,3,4,6,8,10,3,6,7,1,2,3,4,10,3,4,5,7,8,1,6,7,3,4,6,8,4,5,6,7},
    legend pos=south east,
    legend style={font=\footnotesize},
    legend columns=1,
    ymin=0.02,ymax=10000,
    ymajorgrids=true,
    grid style=dashed,
    xticklabel style={font=\tiny},
    clip=false,  
    enlargelimits=false,
    xmin=0,xmax=31
]
\addplot[blue, thick, dashed] coordinates {(1,3600) (2,2093) (3,3600) (4,3600) (5,3601) (6,3600) (7,3600) (8,3600) (9,3600) (10,3600) (11,3600) (12,3600) (13,3600) (14,3600) (15,3600) (16,3601) (17,3600) (18,3600) (19,3600) (20,3600) (21,3600) (22,3600) (23,3600) (24,3600) (25,3601) (26,3600) (27,3600) (28,3600) (29,3601) (30,3600)};
\addlegendentry{\eqref{eq:miqcp-NE}+o}

\addplot[blue, thick] coordinates {(1,6) (2,3147) (3,1) (4,364) (5,9) (6,1) (7,3600) (8,3600) (9,3600) (10,3600) (11,3600) (12,3600) (13,3600) (14,3600) (15,3600) (16,3600) (17,3601) (18,3600) (19,2) (20,3600) (21,3600) (22,3600) (23,3600) (24,3600) (25,3600) (26,3600) (27,3600) (28,3600) (29,3601) (30,3600)};
\addlegendentry{\eqref{eq:miqcp-NE}}

\addplot[red, thick, dashed] coordinates {(1,30) (2,0) (3,12) (4,2) (5,10) (6,6) (7,298) (8,3600) (9,3600) (10,3600) (11,3600) (12,3600) (13,3600) (14,3601) (15,168) (16,141) (17,3600) (18,117) (19,1887) (20,3601) (21,3600) (22,3600) (23,3600) (24,3600) (25,3600) (26,3601) (27,3600) (28,3600) (29,3600) (30,3600)};
\addlegendentry{CMR+o}

\addplot[red, thick] coordinates {(1,5) (2,3600) (3,15) (4,209) (5,1) (6,13) (7,3600) (8,3600) (9,3600) (10,3600) (11,3600) (12,3600) (13,3600) (14,3600) (15,3600) (16,41) (17,3601) (18,3600) (19,3600) (20,3600) (21,3600) (22,3600) (23,3600) (24,3600) (25,3600) (26,3601) (27,3600) (28,3600) (29,3600) (30,3600)};
\addlegendentry{CMR}


\draw[decorate, decoration={brace, mirror, amplitude=5pt}, thick]
  (axis  cs:0.7, \bracescoord) -- (axis  cs:6.3, \bracescoord);

\draw[decorate, decoration={brace, mirror, amplitude=5pt}, thick]
  (axis  cs:6.7, \bracescoord) -- (axis  cs:7.3, \bracescoord);

\draw[decorate, decoration={brace, mirror, amplitude=5pt}, thick]
  (axis  cs:7.7, \bracescoord) -- (axis  cs:9.3, \bracescoord);

\draw[decorate, decoration={brace, mirror, amplitude=5pt}, thick]
  (axis  cs:9.7, \bracescoord) -- (axis  cs:14.3, \bracescoord);

\draw[decorate, decoration={brace, mirror, amplitude=5pt}, thick]
  (axis  cs:14.7, \bracescoord) -- (axis  cs:19.3, \bracescoord);

\draw[decorate, decoration={brace, mirror, amplitude=5pt}, thick]
  (axis  cs:19.7, \bracescoord) -- (axis  cs:22.3, \bracescoord);

\draw[decorate, decoration={brace, mirror, amplitude=5pt}, thick]
  (axis  cs:22.7, \bracescoord) -- (axis  cs:26.3, \bracescoord);

\draw[decorate, decoration={brace, mirror, amplitude=5pt}, thick]
  (axis  cs:26.7, \bracescoord) -- (axis  cs:30.3, \bracescoord);

\node[anchor=north, font=\footnotesize] at (axis cs:3.5, \labelcoord) {5/3};
\node[anchor=north, font=\footnotesize] at (axis cs:6.6, \labelcoord) {5/4};
\node[anchor=north, font=\footnotesize] at (axis cs:9, \labelcoord) {5/5};
\node[anchor=north, font=\footnotesize] at (axis cs:12, \labelcoord) {5/6};
\node[anchor=north, font=\footnotesize] at (axis cs:17, \labelcoord) {6/3};
\node[anchor=north, font=\footnotesize] at (axis cs:21, \labelcoord) {6/4};
\node[anchor=north, font=\footnotesize] at (axis cs:24.5, \labelcoord) {6/5};
\node[anchor=north, font=\footnotesize] at (axis cs:28.5, \labelcoord) {6/6};

\end{semilogyaxis}
\end{tikzpicture}

\vspace{1cm}

\begin{tikzpicture}
\begin{semilogyaxis}[
    title={Runtimes for RoadGraph},
    ylabel={Time (s)},
    width=\fpeval{\xscale*\axisdefaultwidth},
    height=\fpeval{\yscale*\axisdefaultheight},
    xtick={1,2,3,4,5,6,7,8,9,10,11,12,13,14,15,16,17,18,19,20,21,22,23,24,25,26,27,28,29,30},
    xticklabels={1,2,4,5,10,2,8,1,3,5,7,3,5,4,5,6,1,3,9,10,2,3,4,8,3,4,5,6,8,9},
    legend pos=south east,
    legend style={font=\small},
    legend columns=1,
    ymin=0.02,ymax=10000,
    ymajorgrids=true,
    grid style=dashed,
    xticklabel style={font=\tiny},
    clip=false,  
    enlargelimits=false,
    xmin=0,xmax=31
]
\addplot[blue, thick, dashed] coordinates {(1,42) (2,32) (3,1) (4,66) (5,275) (6,3600) (7,3601) (8,3600) (9,3600) (10,3600) (11,3600) (12,3600) (13,3600) (14,3600) (15,3600) (16,3600) (17,3601) (18,3600) (19,3600) (20,3600) (21,3600) (22,3600) (23,3600) (24,3601) (25,3600) (26,3600) (27,3601) (28,3600) (29,3600) (30,3601)};
\addlegendentry{\eqref{eq:miqcp-NE}+o}

\addplot[blue, thick] coordinates {(1,13) (2,1) (3,1) (4,1) (5,1) (6,3600) (7,1) (8,3600) (9,3600) (10,3600) (11,3600) (12,3600) (13,3600) (14,2) (15,3600) (16,3600) (17,3600) (18,3600) (19,3600) (20,3) (21,3600) (22,3600) (23,3600) (24,3600) (25,3600) (26,3600) (27,3601) (28,3600) (29,3600) (30,3601)};
\addlegendentry{\eqref{eq:miqcp-NE}}

\addplot[red, thick, dashed] coordinates {(1,2) (2,0) (3,1) (4,0) (5,1) (6,1) (7,3) (8,7) (9,9) (10,7) (11,182) (12,921) (13,147) (14,2) (15,5) (16,6) (17,107) (18,240) (19,333) (20,49) (21,1540) (22,627) (23,3600) (24,580) (25,2520) (26,3601) (27,3600) (28,3600) (29,3600) (30,3600)};
\addlegendentry{CMR+o}

\addplot[red, thick] coordinates {(1,1) (2,1) (3,0) (4,1) (5,0) (6,1556) (7,2) (8,3600) (9,3600) (10,1926) (11,1324) (12,3600) (13,3600) (14,11) (15,2) (16,3) (17,3600) (18,3600) (19,3600) (20,3075) (21,3600) (22,3600) (23,3600) (24,3600) (25,3601) (26,3600) (27,3600) (28,3601) (29,59) (30,3600)};
\addlegendentry{CMR}

\draw[decorate, decoration={brace, mirror, amplitude=5pt}, thick]
  (axis  cs:0.7, \bracescoord) -- (axis  cs:5.3, \bracescoord);

\draw[decorate, decoration={brace, mirror, amplitude=5pt}, thick]
  (axis  cs:5.7, \bracescoord) -- (axis  cs:7.3, \bracescoord);

\draw[decorate, decoration={brace, mirror, amplitude=5pt}, thick]
  (axis  cs:7.7, \bracescoord) -- (axis  cs:11.3, \bracescoord);

\draw[decorate, decoration={brace, mirror, amplitude=5pt}, thick]
  (axis  cs:11.7, \bracescoord) -- (axis  cs:13.3, \bracescoord);

\draw[decorate, decoration={brace, mirror, amplitude=5pt}, thick]
  (axis  cs:13.7, \bracescoord) -- (axis  cs:16.3, \bracescoord);

\draw[decorate, decoration={brace, mirror, amplitude=5pt}, thick]
  (axis  cs:16.7, \bracescoord) -- (axis  cs:20.3, \bracescoord);

\draw[decorate, decoration={brace, mirror, amplitude=5pt}, thick]
  (axis  cs:20.7, \bracescoord) -- (axis  cs:24.3, \bracescoord);

\draw[decorate, decoration={brace, mirror, amplitude=5pt}, thick]
  (axis  cs:24.7, \bracescoord) -- (axis  cs:30.3, \bracescoord);

\node[anchor=north, font=\footnotesize] at (axis cs:3, \labelcoord) {5/3};
\node[anchor=north, font=\footnotesize] at (axis cs:6.5, \labelcoord) {5/4};
\node[anchor=north, font=\footnotesize] at (axis cs:9.5, \labelcoord) {5/5};
\node[anchor=north, font=\footnotesize] at (axis cs:12.5, \labelcoord) {5/6};
\node[anchor=north, font=\footnotesize] at (axis cs:15, \labelcoord) {6/3};
\node[anchor=north, font=\footnotesize] at (axis cs:18.5, \labelcoord) {6/4};
\node[anchor=north, font=\footnotesize] at (axis cs:22.5, \labelcoord) {6/5};
\node[anchor=north, font=\footnotesize] at (axis cs:27.5, \labelcoord) {6/6};

\end{semilogyaxis}
\end{tikzpicture}
\end{figure*}

\begin{figure*}[h!]
\centering
\begin{tikzpicture}
\begin{semilogyaxis}[
    title={Runtimes for SmallWorldGraph},
    width=\fpeval{\xscale*\axisdefaultwidth},
    height=\fpeval{\yscale*\axisdefaultheight},
    ylabel={Time (s)},
    xtick={1,2,3,4,5,6,7,8,9,10,11,12,13,14,15,16,17,18,19,20,21,22,23,24,25,26,27,28,29,30,31,32,33,34,35,36},
    xticklabels={6,14,21,24,26,27,28,29,30,5,11,12,17,20,21,22,27,1,8,9,12,14,15,21,22,25,1,2,3,7,12,17,20,22,24,28},
    legend pos=south east,
    legend style={font=\small},
    legend columns=1,
    ymin=0.02,ymax=10000,
    ymajorgrids=true,
    grid style=dashed,
    xticklabel style={font=\tiny},
    clip=false,  
    enlargelimits=false,
    xmin=0,xmax=37
]
\addplot[blue, thick, dashed] coordinates {(1,3600) (2,3600) (3,3600) (4,3600) (5,3600) (6,3600) (7,3600) (8,3601) (9,3600) (10,3600) (11,3600) (12,3600) (13,3600) (14,3600) (15,3600) (16,3600) (17,3601) (18,3600) (19,3600) (20,3601) (21,3600) (22,3600) (23,3600) (24,3601) (25,3600) (26,3600) (27,3600) (28,3600) (29,3600) (30,3600) (31,3600) (32,3601) (33,3600) (34,3601) (35,3600) (36,3601)};
\addlegendentry{\eqref{eq:miqcp-NE}+o}

\addplot[blue, thick] coordinates {(1,1) (2,3600) (3,2262) (4,1) (5,3600) (6,3600) (7,3600) (8,2) (9,3600) (10,3600) (11,3600) (12,3600) (13,3600) (14,3600) (15,3600) (16,3600) (17,3600) (18,3600) (19,3600) (20,3600) (21,3600) (22,3600) (23,3600) (24,3600) (25,3600) (26,3600) (27,3600) (28,3601) (29,3600) (30,3600) (31,3600) (32,3601) (33,3600) (34,3600) (35,3601) (36,3601)};
\addlegendentry{\eqref{eq:miqcp-NE}}

\addplot[red, thick, dashed] coordinates {(1,87) (2,1047) (3,42) (4,13) (5,2308) (6,4) (7,574) (8,26) (9,106) (10,3600) (11,3601) (12,3600) (13,3600) (14,3600) (15,3600) (16,3601) (17,3600) (18,3600) (19,3600) (20,3600) (21,3601) (22,3600) (23,3600) (24,3600) (25,3601) (26,3600) (27,3601) (28,3600) (29,3601) (30,3600) (31,3600) (32,3600) (33,3600) (34,3601) (35,3600) (36,3600)};
\addlegendentry{CMR+o}

\addplot[red, thick] coordinates {(1,3600) (2,3600) (3,3600) (4,3600) (5,3601) (6,3600) (7,3600) (8,3600) (9,3600) (10,3601) (11,3600) (12,3600) (13,3600) (14,3601) (15,3600) (16,3600) (17,3600) (18,3600) (19,3601) (20,3600) (21,3600) (22,3601) (23,3600) (24,3600) (25,3601) (26,3600) (27,3601) (28,3600) (29,3601) (30,3600) (31,3600) (32,3600) (33,3601) (34,3600) (35,3600) (36,3600)};
\addlegendentry{CMR}

\draw[decorate, decoration={brace, mirror, amplitude=5pt}, thick]
  (axis  cs:0.7, \bracescoord) -- (axis  cs:9.3, \bracescoord);

\draw[decorate, decoration={brace, mirror, amplitude=5pt}, thick]
  (axis  cs:9.7, \bracescoord) -- (axis  cs:17.3, \bracescoord);

\draw[decorate, decoration={brace, mirror, amplitude=5pt}, thick]
  (axis  cs:17.7, \bracescoord) -- (axis  cs:26.3, \bracescoord);

\draw[decorate, decoration={brace, mirror, amplitude=5pt}, thick]
  (axis  cs:26.7, \bracescoord) -- (axis  cs:36.3, \bracescoord);

\node[anchor=north, font=\footnotesize] at (axis cs:5, \labelcoord) {6/3};
\node[anchor=north, font=\footnotesize] at (axis cs:13.5, \labelcoord) {6/4};
\node[anchor=north, font=\footnotesize] at (axis cs:22, \labelcoord) {6/5};
\node[anchor=north, font=\footnotesize] at (axis cs:31.5, \labelcoord) {6/6};

\end{semilogyaxis}
\end{tikzpicture}

\end{figure*}

%% file: figs/runtimes_random.tex
\def\bracescoord{0.008}
\def\labelcoord{0.005}
\def\bracescoordapp{0.000000000011}
\def\labelcoordapp{0.0000000000044}
\def\xscale{1.3}
\def\yscale{1.0}
\def\btwspace{.3cm}
\begin{figure*}[h!]
\centering
\begin{tikzpicture}
\begin{semilogyaxis}[
    title={Runtimes for random games},
    ylabel={Time (s)},
    width=\fpeval{\xscale*\axisdefaultwidth},
    height=\fpeval{\yscale*\axisdefaultheight},
    xtick={1,2,3,4,5,6,7,8,9,10,11,12,13,14,15,16,17,18,19,20,21,22,23,24,25,26,27,28,29,30,31,32,33,34,35,36,37,38,39,40,41,42,43,44,45,46,47,48,49,50,51,52,53,54,55,56,57,58,59,60,61,62,63,64,65,66,67,68,69,70,71,72,73,74,75,76,77,78,79,80,81,82,83,84,85,86,87,88,89,90,91,92,93,94,95,96,97,98,99,100,101,102,103,104,105,106,107,108,109,110,111,112,113,114,115,116,117,118,119,120,121,122,123,124,125,126,127,128,129,130,131,132,133,134,135,136,137,138,139,140,141,142,143,144,145,146,147,148,149,150,151,152,153,154,155,156,157,158,159,160},
    xticklabels={},
    legend pos=south east,
    legend style={font=\footnotesize},
    legend columns=1,
    ymin=0.02,ymax=10000,
    ymajorgrids=true,
    grid style=dashed,
    xticklabel style={font=\tiny},
    clip=false,  
    enlargelimits=false,
    xmin=50,xmax=162
]
\addplot[blue, thick, dashed] coordinates {(51,3) (52,2) (53,2) (54,2) (55,3) (56,2) (57,2) (58,2) (59,2) (60,2) (61,555) (62,558) (63,557) (64,554) (65,551) (66,553) (67,552) (68,555) (69,554) (70,552) (71,3600) (72,3600) (73,3600) (74,3600) (75,3600) (76,3601) (77,3600) (78,3600) (79,3600) (80,3600) (81,4) (82,5) (83,5) (84,5) (85,6) (86,6) (87,6) (88,6) (89,6) (90,5) (91,3600) (92,3600) (93,3600) (94,3600) (95,3600) (96,3601) (97,3600) (98,3600) (99,3600) (100,3600) (101,3600) (102,3600) (103,3600) (104,3600) (105,3600) (106,3600) (107,3600) (108,3600) (109,3600) (110,3600) (111,3600) (112,3600) (113,3600) (114,3600) (115,3600) (116,3600) (117,3600) (118,3600) (119,3600) (120,3600) (121,12) (122,16) (123,16) (124,16) (125,16) (126,16) (127,16) (128,17) (129,16) (130,17) (131,3600) (132,3601) (133,3600) (134,3600) (135,3600) (136,3600) (137,3601) (138,3600) (139,3600) (140,3600) (141,3600) (142,3600) (143,3600) (144,3600) (145,3600) (146,3600) (147,3600) (148,3600) (149,3600) (150,3600) (151,3600) (152,3600) (153,3600) (154,3600) (155,3601) (156,3600) (157,3601) (158,3600) (159,3601) (160,3600)};
\addlegendentry{\eqref{eq:miqcp-NE}-o}

\addplot[blue, thick] coordinates {
(51,1) (52,1) (53,1) (54,1) (55,1) (56,1) (57,1) (58,1) (59,2) (60,1) (61,3600) (62,3600) (63,3600) (64,3600) (65,3600) (66,3600) (67,3600) (68,3600) (69,3600) (70,3600) (71,3600) (72,3600) (73,3600) (74,3600) (75,3600) (76,3600) (77,3600) (78,3600) (79,3600) (80,3600) (81,8) (82,10) (83,10) (84,9) (85,9) (86,9) (87,11) (88,11) (89,11) (90,8) (91,3600) (92,3600) (93,3600) (94,3600) (95,3600) (96,3600) (97,3600) (98,3600) (99,3600) (100,3600) (101,3600) (102,3600) (103,3600) (104,3600) (105,3600) (106,3600) (107,3600) (108,3600) (109,3600) (110,3600) (111,3600) (112,3600) (113,3600) (114,3600) (115,3600) (116,3600) (117,3600) (118,3600) (119,3600) (120,3600) (121,3600) (122,3600) (123,3600) (124,3601) (125,3600) (126,3600) (127,3600) (128,3600) (129,3600) (130,3600) (131,3600) (132,3600) (133,3600) (134,3600) (135,3600) (136,3600) (137,3600) (138,3600) (139,3600) (140,3600) (141,3600) (142,3600) (143,3600) (144,3600) (145,3600) (146,3600) (147,3600) (148,3600) (149,3600) (150,3600) (151,3601) (152,3601) (153,3600) (154,3600) (155,3600) (156,3600) (157,3600) (158,3600) (159,3600) (160,3600)};
\addlegendentry{CMR}

\addplot[red, thick, dashed] coordinates {
(51,6) (52,6) (53,6) (54,5) (55,6) (56,5) (57,6) (58,5) (59,5) (60,5) (61,7) (62,6) (63,6) (64,6) (65,6) (66,6) (67,6) (68,7) (69,6) (70,6) (71,6) (72,6) (73,7) (74,7) (75,6) (76,7) (77,6) (78,7) (79,6) (80,7) (81,7) (82,7) (83,6) (84,6) (85,6) (86,7) (87,6) (88,6) (89,6) (90,6) (91,7) (92,7) (93,7) (94,7) (95,7) (96,7) (97,7) (98,7) (99,7) (100,7) (101,8) (102,8) (103,8) (104,7) (105,8) (106,8) (107,7) (108,8) (109,8) (110,7) (111,8) (112,8) (113,9) (114,8) (115,9) (116,8) (117,8) (118,9) (119,8) (120,9) (121,7) (122,7) (123,8) (124,7) (125,8) (126,7) (127,7) (128,8) (129,7) (130,8) (131,8) (132,8) (133,8) (134,9) (135,8) (136,8) (137,8) (138,9) (139,8) (140,9) (141,10) (142,9) (143,10) (144,9) (145,10) (146,9) (147,9) (148,10) (149,9) (150,9) (151,11) (152,11) (153,11) (154,11) (155,11) (156,12) (157,11) (158,11) (159,11) (160,11)};
\addlegendentry{ADIDAS}

\addplot[black, thick,dashed] coordinates {
(51,0.168393) (52,0.168002) (53,0.168312) (54,0.168189) (55,0.168075) (56,0.168979) (57,0.169133) (58,0.168151) (59,0.1684) (60,0.168778) (61,1.00395) (62,0.893413) (63,0.906446) (64,0.940936) (65,0.890867) (66,0.889964) (67,0.999594) (68,0.903363) (69,0.888775) (70,1.01522) (71,1.57223) (72,1.54937) (73,1.55008) (74,1.5516) (75,1.55068) (76,1.55596) (77,1.55492) (78,1.55369) (79,1.55363) (80,1.55069) (81,0.1253) (82,0.125221) (83,0.125055) (84,0.13074) (85,0.125168) (86,0.125207) (87,0.124911) (88,0.125677) (89,0.12568) (90,0.126663) (91,0.540566) (92,0.54174) (93,0.541161) (94,0.541201) (95,0.545978) (96,0.542279) (97,0.541122) (98,0.540894) (99,0.541886) (100,0.542358) (101,4.29663) (102,4.31626) (103,4.29294) (104,4.30375) (105,4.29241) (106,4.30055) (107,4.29858) (108,4.29397) (109,4.29938) (110,4.28962) (111,5.41146) (112,5.42835) (113,5.42464) (114,5.41725) (115,5.43384) (116,5.42751) (117,5.43636) (118,5.4206) (119,5.43119) (120,5.43229) (121,0.918798) (122,0.915868) (123,0.914866) (124,0.915326) (125,0.916965) (126,0.916524) (127,0.914062) (128,0.914809) (129,0.916462) (130,0.915977) (131,7.68134) (132,7.69354) (133,7.70165) (134,7.69514) (135,7.67589) (136,7.68635) (137,7.70425) (138,7.70569) (139,7.68407) (140,7.68247) (141,52.4573) (142,52.7014) (143,52.5925) (144,52.7824) (145,52.6637) (146,52.7327) (147,52.7553) (148,52.6206) (149,52.6639) (150,52.6893) (151,237.4) (152,237.459) (153,237.331) (154,237.573) (155,237.684) (156,237.379) (157,237.721) (158,237.449) (159,237.25) (160,237.465)};
\addlegendentry{SBnB}

\addplot[black, thick] coordinates {
(51,0.041461) (52,0.041439) (53,0.04188) (54,0.041738) (55,0.041902) (56,0.041361) (57,0.04164) (58,0.041426) (59,0.04165) (60,0.041704) (61,0.173897) (62,0.174338) (63,0.173578) (64,0.173221) (65,0.174052) (66,0.173909) (67,0.173817) (68,0.174142) (69,0.17398) (70,0.173389) (71,0.317045) (72,0.31819) (73,0.3174) (74,0.316545) (75,0.3181) (76,0.316511) (77,0.3167) (78,0.316246) (79,0.317077) (80,0.316906) (81,0.092067) (82,0.091685) (83,0.091852) (84,0.093086) (85,0.092079) (86,0.091804) (87,0.091996) (88,0.091827) (89,0.092191) (90,0.091847) (91,0.286201) (92,0.284875) (93,0.28468) (94,0.28553) (95,0.284616) (96,0.284873) (97,0.285028) (98,0.284861) (99,0.284794) (100,0.285274) (101,2.26599) (102,2.27209) (103,2.26664) (104,2.2686) (105,2.26752) (106,2.27786) (107,2.26973) (108,2.27231) (109,2.26562) (110,2.27555) (111,3.9277) (112,3.92097) (113,3.9287) (114,3.92062) (115,3.91993) (116,3.92599) (117,3.92459) (118,3.90917) (119,3.91509) (120,3.92038) (121,0.739865) (122,0.739207) (123,0.74088) (124,0.739972) (125,0.738274) (126,0.740787) (127,0.737915) (128,0.741338) (129,0.740117) (130,0.739513) (131,3.93216) (132,3.93204) (133,3.90741) (134,3.93769) (135,3.93706) (136,3.93462) (137,3.9312) (138,3.92392) (139,3.93229) (140,3.93659) (141,18.9423) (142,18.898) (143,18.9113) (144,18.9072) (145,18.911) (146,18.94) (147,18.9072) (148,18.914) (149,18.8954) (150,18.911) (151,62.0214) (152,61.998) (153,62.3778) (154,62.4237) (155,62.4715) (156,62.0005) (157,62.4226) (158,62.0352) (159,61.975) (160,62.3406)};
\addlegendentry{SBnB-e}


\draw[decorate, decoration={brace, mirror, amplitude=5pt}, thick]
  (axis  cs:51.7, \bracescoord) -- (axis  cs:60.3, \bracescoord);

\draw[decorate, decoration={brace, mirror, amplitude=5pt}, thick]
  (axis  cs:61.7, \bracescoord) -- (axis  cs:70.3, \bracescoord);

\draw[decorate, decoration={brace, mirror, amplitude=5pt}, thick]
  (axis  cs:71.7, \bracescoord) -- (axis  cs:80.3, \bracescoord);

\draw[decorate, decoration={brace, mirror, amplitude=5pt}, thick]
  (axis  cs:81.7, \bracescoord) -- (axis  cs:90.3, \bracescoord);

\draw[decorate, decoration={brace, mirror, amplitude=5pt}, thick]
  (axis  cs:91.7, \bracescoord) -- (axis  cs:100.3, \bracescoord);

\draw[decorate, decoration={brace, mirror, amplitude=5pt}, thick]
  (axis  cs:101.7, \bracescoord) -- (axis  cs:110.3, \bracescoord);

\draw[decorate, decoration={brace, mirror, amplitude=5pt}, thick]
  (axis  cs:111.7, \bracescoord) -- (axis  cs:120.3, \bracescoord);

\draw[decorate, decoration={brace, mirror, amplitude=5pt}, thick]
  (axis  cs:121.7, \bracescoord) -- (axis  cs:130.3, \bracescoord);

\draw[decorate, decoration={brace, mirror, amplitude=5pt}, thick]
  (axis  cs:131.7, \bracescoord) -- (axis  cs:140.3, \bracescoord);

\draw[decorate, decoration={brace, mirror, amplitude=5pt}, thick]
  (axis  cs:141.7, \bracescoord) -- (axis  cs:150.3, \bracescoord);

\draw[decorate, decoration={brace, mirror, amplitude=5pt}, thick]
  (axis  cs:151.7, \bracescoord) -- (axis  cs:160.3, \bracescoord);

\node[anchor=north, font=\footnotesize] at (axis cs:55.5, \labelcoord) {4/4};
\node[anchor=north, font=\footnotesize] at (axis cs:65.5, \labelcoord) {4/5};
\node[anchor=north, font=\footnotesize] at (axis cs:75.5, \labelcoord) {4/6};
\node[anchor=north, font=\footnotesize] at (axis cs:85.5, \labelcoord) {5/3};
\node[anchor=north, font=\footnotesize] at (axis cs:95.5, \labelcoord) {5/4};
\node[anchor=north, font=\footnotesize] at (axis cs:105.5, \labelcoord) {5/5};
\node[anchor=north, font=\footnotesize] at (axis cs:115.5, \labelcoord) {5/6};
\node[anchor=north, font=\footnotesize] at (axis cs:125.5, \labelcoord) {6/3};
\node[anchor=north, font=\footnotesize] at (axis cs:135.5, \labelcoord) {6/4};
\node[anchor=north, font=\footnotesize] at (axis cs:145.5, \labelcoord) {6/5};
\node[anchor=north, font=\footnotesize] at (axis cs:155.5, \labelcoord) {6/6};

\end{semilogyaxis}
\end{tikzpicture}
\end{figure*}

%% file: figs/eps_random.tex
\def\bracescoord{0.008}
\def\labelcoord{0.005}
\def\bracescoordapp{0.000000000011}
\def\labelcoordapp{0.0000000000044}
\def\xscale{1.3}
\def\yscale{1.0}
\def\btwspace{.3cm}
\begin{figure*}[h!]
\centering
\begin{tikzpicture}
\begin{semilogyaxis}[
    title={Approximation quality for random games},
    ylabel={Time (s)},
    width=\fpeval{\xscale*\axisdefaultwidth},
    height=\fpeval{\yscale*\axisdefaultheight},
    xtick={1,2,3,4,5,6,7,8,9,10,11,12,13,14,15,16,17,18,19,20,21,22,23,24,25,26,27,28,29,30,31,32,33,34,35,36,37,38,39,40,41,42,43,44,45,46,47,48,49,50,51,52,53,54,55,56,57,58,59,60,61,62,63,64,65,66,67,68,69,70,71,72,73,74,75,76,77,78,79,80,81,82,83,84,85,86,87,88,89,90,91,92,93,94,95,96,97,98,99,100,101,102,103,104,105,106,107,108,109,110,111,112,113,114,115,116,117,118,119,120,121,122,123,124,125,126,127,128,129,130,131,132,133,134,135,136,137,138,139,140,141,142,143,144,145,146,147,148,149,150,151,152,153,154,155,156,157,158,159,160},
    xticklabels={},
    legend pos=south east,
    legend style={font=\footnotesize},
    legend columns=1,
    ymin=5e-11,ymax=7e-1,
    ymajorgrids=true,
    grid style=dashed,
    xticklabel style={font=\tiny},
    clip=false,  
    enlargelimits=false,
    xmin=50,xmax=162
]

\addplot[red, thick, dashed] coordinates {
(51,5.558583e-02) (52,5.558583e-02) (53,5.558583e-02) (54,5.558583e-02) (55,5.558583e-02) (56,5.558583e-02) (57,5.558583e-02) (58,5.558583e-02) (59,5.558583e-02) (60,5.558583e-02) (61,3.853612e-02) (62,3.853612e-02) (63,3.853612e-02) (64,3.853612e-02) (65,3.853612e-02) (66,3.853612e-02) (67,3.853612e-02) (68,3.853612e-02) (69,3.853612e-02) (70,3.853612e-02) (71,5.385407e-02) (72,5.385407e-02) (73,5.385407e-02) (74,5.385407e-02) (75,5.385407e-02) (76,5.385407e-02) (77,5.385407e-02) (78,5.385407e-02) (79,5.385407e-02) (80,5.385407e-02) (81,3.633724e-02) (82,3.633724e-02) (83,3.633724e-02) (84,3.633724e-02) (85,3.633724e-02) (86,3.633724e-02) (87,3.633724e-02) (88,3.633724e-02) (89,3.633724e-02) (90,3.633724e-02) (91,4.831044e-02) (92,4.831044e-02) (93,4.831044e-02) (94,4.831044e-02) (95,4.831044e-02) (96,4.831044e-02) (97,4.831044e-02) (98,4.831044e-02) (99,4.831044e-02) (100,4.831044e-02) (101,1.754719e-02) (102,1.754719e-02) (103,1.754719e-02) (104,1.754719e-02) (105,1.754719e-02) (106,1.754719e-02) (107,1.754719e-02) (108,1.754719e-02) (109,1.754719e-02) (110,1.754719e-02) (111,1.269740e-02) (112,1.269740e-02) (113,1.269740e-02) (114,1.269740e-02) (115,1.269740e-02) (116,1.269740e-02) (117,1.269740e-02) (118,1.269740e-02) (119,1.269740e-02) (120,1.269740e-02) (121,5.209181e-02) (122,5.209181e-02) (123,5.209181e-02) (124,5.209181e-02) (125,5.209181e-02) (126,5.209181e-02) (127,5.209181e-02) (128,5.209181e-02) (129,5.209181e-02) (130,5.209181e-02) (131,1.934691e-02) (132,1.934691e-02) (133,1.934691e-02) (134,1.934691e-02) (135,1.934691e-02) (136,1.934691e-02) (137,1.934691e-02) (138,1.934691e-02) (139,1.934691e-02) (140,1.934691e-02) (141,6.213119e-03) (142,6.213119e-03) (143,6.213119e-03) (144,6.213119e-03) (145,6.213119e-03) (146,6.213119e-03) (147,6.213119e-03) (148,6.213119e-03) (149,6.213119e-03) (150,6.213119e-03) (151,7.846471e-03) (152,7.846471e-03) (153,7.846471e-03) (154,7.846471e-03) (155,7.846471e-03) (156,7.846471e-03) (157,7.846471e-03) (158,7.846471e-03) (159,7.846471e-03) (160,7.846471e-03)};
\addlegendentry{ADIDAS}

\addplot[black, thick,dashed] coordinates {
(51,1.50917e-10) (52,1.50917e-10) (53,1.50917e-10) (54,1.50917e-10) (55,1.50917e-10) (56,1.50917e-10) (57,1.50917e-10) (58,1.50917e-10) (59,1.50917e-10) (60,1.50917e-10) (61,7.83114e-7) (62,7.83114e-7) (63,7.83114e-7) (64,7.83114e-7) (65,7.83114e-7) (66,7.83114e-7) (67,7.83114e-7) (68,7.83114e-7) (69,7.83114e-7) (70,7.83114e-7) (71,1.56483e-7) (72,1.56483e-7) (73,1.56483e-7) (74,1.56483e-7) (75,1.56483e-7) (76,1.56483e-7) (77,1.56483e-7) (78,1.56483e-7) (79,1.56483e-7) (80,1.56483e-7) (81,1.40417e-7) (82,1.40417e-7) (83,1.40417e-7) (84,1.40417e-7) (85,1.40417e-7) (86,1.40417e-7) (87,1.40417e-7) (88,1.40417e-7) (89,1.40417e-7) (90,1.40417e-7) (91,1.19683e-10) (92,1.19683e-10) (93,1.19683e-10) (94,1.19683e-10) (95,1.19683e-10) (96,1.19683e-10) (97,1.19683e-10) (98,1.19683e-10) (99,1.19683e-10) (100,1.19683e-10) (101,5.16617e-10) (102,5.16617e-10) (103,5.16617e-10) (104,5.16617e-10) (105,5.16617e-10) (106,5.16617e-10) (107,5.16617e-10) (108,5.16617e-10) (109,5.16617e-10) (110,5.16617e-10) (111,6.42416e-9) (112,6.42416e-9) (113,6.42416e-9) (114,6.42416e-9) (115,6.42416e-9) (116,6.42416e-9) (117,6.42416e-9) (118,6.42416e-9) (119,6.42416e-9) (120,6.42416e-9) (121,3.06925e-9) (122,3.06925e-9) (123,3.06925e-9) (124,3.06925e-9) (125,3.06925e-9) (126,3.06925e-9) (127,3.06925e-9) (128,3.06925e-9) (129,3.06925e-9) (130,3.06925e-9) (131,2.43004e-9) (132,2.43004e-9) (133,2.43004e-9) (134,2.43004e-9) (135,2.43004e-9) (136,2.43004e-9) (137,2.43004e-9) (138,2.43004e-9) (139,2.43004e-9) (140,2.43004e-9) (141,8.1876e-9) (142,8.1876e-9) (143,8.1876e-9) (144,8.1876e-9) (145,8.1876e-9) (146,8.1876e-9) (147,8.1876e-9) (148,8.1876e-9) (149,8.1876e-9) (150,8.1876e-9) (151,2.52657e-7) (152,2.52657e-7) (153,2.52657e-7) (154,2.52657e-7) (155,2.52657e-7) (156,2.52657e-7) (157,2.52657e-7) (158,2.52657e-7) (159,2.52657e-7) (160,2.52657e-7)};
\addlegendentry{SBnB}

\addplot[black, thick] coordinates {
(51,0.00230822) (52,0.00230822) (53,0.00230822) (54,0.00230822) (55,0.00230822) (56,0.00230822) (57,0.00230822) (58,0.00230822) (59,0.00230822) (60,0.00230822) (61,0.000777769) (62,0.000777769) (63,0.000777769) (64,0.000777769) (65,0.000777769) (66,0.000777769) (67,0.000777769) (68,0.000777769) (69,0.000777769) (70,0.000777769) (71,1.56483e-7) (72,1.56483e-7) (73,1.56483e-7) (74,1.56483e-7) (75,1.56483e-7) (76,1.56483e-7) (77,1.56483e-7) (78,1.56483e-7) (79,1.56483e-7) (80,1.56483e-7) (81,6.46975e-9) (82,6.46975e-9) (83,6.46975e-9) (84,6.46975e-9) (85,6.46975e-9) (86,6.46975e-9) (87,6.46975e-9) (88,6.46975e-9) (89,6.46975e-9) (90,6.46975e-9) (91,1.13756e-7) (92,1.13756e-7) (93,1.13756e-7) (94,1.13756e-7) (95,1.13756e-7) (96,1.13756e-7) (97,1.13756e-7) (98,1.13756e-7) (99,1.13756e-7) (100,1.13756e-7) (101,0.000429462) (102,0.000429462) (103,0.000429462) (104,0.000429462) (105,0.000429462) (106,0.000429462) (107,0.000429462) (108,0.000429462) (109,0.000429462) (110,0.000429462) (111,1.45855e-6) (112,1.45855e-6) (113,1.45855e-6) (114,1.45855e-6) (115,1.45855e-6) (116,1.45855e-6) (117,1.45855e-6) (118,1.45855e-6) (119,1.45855e-6) (120,1.45855e-6) (121,2.32675e-7) (122,2.32675e-7) (123,2.32675e-7) (124,2.32675e-7) (125,2.32675e-7) (126,2.32675e-7) (127,2.32675e-7) (128,2.32675e-7) (129,2.32675e-7) (130,2.32675e-7) (131,5.17897e-5) (132,5.17897e-5) (133,5.17897e-5) (134,5.17897e-5) (135,5.17897e-5) (136,5.17897e-5) (137,5.17897e-5) (138,5.17897e-5) (139,5.17897e-5) (140,5.17897e-5) (141,4.66961e-5) (142,4.66961e-5) (143,4.66961e-5) (144,4.66961e-5) (145,4.66961e-5) (146,4.66961e-5) (147,4.66961e-5) (148,4.66961e-5) (149,4.66961e-5) (150,4.66961e-5) (151,2.52657e-7) (152,2.52657e-7) (153,2.52657e-7) (154,2.52657e-7) (155,2.52657e-7) (156,2.52657e-7) (157,2.52657e-7) (158,2.52657e-7) (159,2.52657e-7) (160,2.52657e-7)};
\addlegendentry{SBnB-e}


\draw[decorate, decoration={brace, mirror, amplitude=5pt}, thick]
  (axis  cs:51.7, \bracescoordapp) -- (axis  cs:60.3, \bracescoordapp);

\draw[decorate, decoration={brace, mirror, amplitude=5pt}, thick]
  (axis  cs:61.7, \bracescoordapp) -- (axis  cs:70.3, \bracescoordapp);

\draw[decorate, decoration={brace, mirror, amplitude=5pt}, thick]
  (axis  cs:71.7, \bracescoordapp) -- (axis  cs:80.3, \bracescoordapp);

\draw[decorate, decoration={brace, mirror, amplitude=5pt}, thick]
  (axis  cs:81.7, \bracescoordapp) -- (axis  cs:90.3, \bracescoordapp);

\draw[decorate, decoration={brace, mirror, amplitude=5pt}, thick]
  (axis  cs:91.7, \bracescoordapp) -- (axis  cs:100.3, \bracescoordapp);

\draw[decorate, decoration={brace, mirror, amplitude=5pt}, thick]
  (axis  cs:101.7, \bracescoordapp) -- (axis  cs:110.3, \bracescoordapp);

\draw[decorate, decoration={brace, mirror, amplitude=5pt}, thick]
  (axis  cs:111.7, \bracescoordapp) -- (axis  cs:120.3, \bracescoordapp);

\draw[decorate, decoration={brace, mirror, amplitude=5pt}, thick]
  (axis  cs:121.7, \bracescoordapp) -- (axis  cs:130.3, \bracescoordapp);

\draw[decorate, decoration={brace, mirror, amplitude=5pt}, thick]
  (axis  cs:131.7, \bracescoordapp) -- (axis  cs:140.3, \bracescoordapp);

\draw[decorate, decoration={brace, mirror, amplitude=5pt}, thick]
  (axis  cs:141.7, \bracescoordapp) -- (axis  cs:150.3, \bracescoordapp);

\draw[decorate, decoration={brace, mirror, amplitude=5pt}, thick]
  (axis  cs:151.7, \bracescoordapp) -- (axis  cs:160.3, \bracescoordapp);

\node[anchor=north, font=\footnotesize] at (axis cs:55.5, \labelcoordapp) {4/4};
\node[anchor=north, font=\footnotesize] at (axis cs:65.5, \labelcoordapp) {4/5};
\node[anchor=north, font=\footnotesize] at (axis cs:75.5, \labelcoordapp) {4/6};
\node[anchor=north, font=\footnotesize] at (axis cs:85.5, \labelcoordapp) {5/3};
\node[anchor=north, font=\footnotesize] at (axis cs:95.5, \labelcoordapp) {5/4};
\node[anchor=north, font=\footnotesize] at (axis cs:105.5, \labelcoordapp) {5/5};
\node[anchor=north, font=\footnotesize] at (axis cs:115.5, \labelcoordapp) {5/6};
\node[anchor=north, font=\footnotesize] at (axis cs:125.5, \labelcoordapp) {6/3};
\node[anchor=north, font=\footnotesize] at (axis cs:135.5, \labelcoordapp) {6/4};
\node[anchor=north, font=\footnotesize] at (axis cs:145.5, \labelcoordapp) {6/5};
\node[anchor=north, font=\footnotesize] at (axis cs:155.5, \labelcoordapp) {6/6};

\end{semilogyaxis}
\end{tikzpicture}
\end{figure*}